\documentclass[copyright,creativecommons]{eptcs}
\providecommand{\event}{AiML 2026} 

\usepackage{aiml26}

\usepackage{iftex}

\usepackage{amsmath,amssymb,mathtools}
\usepackage{stmaryrd}
\usepackage{booktabs}
\ifpdf
\usepackage{underscore}         
\usepackage[T1]{fontenc}        
\else
\usepackage{breakurl}           
\fi
\newcommand{\K}{\ensuremath{\mathsf{K}}}

\newcommand{\Ex}{\ensuremath{\mathcal{E}}}

\newcommand{\SCUzero}{\ensuremath{\mathsf{SCU}_{\kappa}}}
\newcommand{\SCUi}{\ensuremath{\mathsf{SCU}}}
\newcommand{\UO}{\ensuremath{\mathsf{U}}}

\newcommand{\Kg}{\ensuremath{\mathtt{DISTK}}}
\newcommand{\Kgg}{\ensuremath{\mathtt{DISTU}}}
\newcommand{\T}{\ensuremath{\mathtt{T}}}

\newcommand{\IYB}{\ensuremath{\mathtt{DEG}}}
\newcommand{\UYC}{\ensuremath{\mathtt{UYC}}}
\newcommand{\CMP}{\ensuremath{\mathtt{CMP}}}
\newcommand{\CMPn}{\ensuremath{\mathtt{CMPn}}}

\newcommand{\NE}{\ensuremath{\mathtt{NE}}}

\newcommand{\Kyy}{\ensuremath{\mathsf{Ky}}}

\newcommand{\M}{\ensuremath{\mathcal{M}}}

\newcommand{\ELBU}{\ensuremath{\mathbf{CELU}}}

\newcommand{\TAUT}{\ensuremath{\mathtt{TAUT}}}
\newcommand{\4}{\ensuremath{\mathtt{4}}}
\newcommand{\5}{\ensuremath{\mathtt{5}}}
\newcommand{\IMP}{\ensuremath{\mathtt{UYK}}}

\newcommand{\Asym}{\ensuremath{\mathtt{ASYM}}}
\newcommand{\Trans}{\ensuremath{\mathtt{TRANS}}}
\newcommand{\Irref}{\ensuremath{\mathtt{IRREF}}}

\newcommand{\gr}{\ensuremath{\mathrm{gr}}}

	\newcommand{\KYU}{\ensuremath{\mathtt{KYU}}}
	\newcommand{\MP}{\ensuremath{\mathtt{MP}}}
	\newcommand{\Nn}{\ensuremath{\mathtt{N}}}
	\newcommand{\la}{\langle}
	\newcommand{\ra}{\rangle}
	
	\renewcommand{\phi}{\varphi}

	\makeatletter
	\newcommand*\bigcdot{\mathpalette\bigcdot@{.5}}
	\newcommand*\bigcdot@[2]{\mathbin{\vcenter{\hbox{\scalebox{#2}{$\m@th#1\bullet$}}}}}
	\makeatother
	
	\title{Better Understanding, Understanding Better}
	\author{Yu Wei
		\institute{Department of Philosophy\\
			East China Normal University
			\\
			Shanghai, China}
		\email{ywei@philo.ecnu.edu.cn}
	}

	\newcommand{\titlerunning}{Better Understanding, Understanding Better}
	\newcommand{\authorrunning}{Y. Wei}
	
	\hypersetup{
		bookmarksnumbered,
		pdftitle    = {\titlerunning},
		pdfauthor   = {\authorrunning},
		pdfsubject  = {EPTCS},               
		pdfkeywords = {Understanding why, Knowing why, Graded modalities, Comparative modal logic, Justification logic} 
	}
	
\begin{document}
		\maketitle
		\begin{abstract}
			``Any fool can know; the point is to understand.'' A well-known remark often attributed to Einstein captures a widely shared intuition: understanding is more than merely knowing. Yet epistemic logic has paid relatively little attention to understanding, despite its central role in contemporary epistemology, philosophy of science, and recent debates about AI. A recurring theme in the philosophical literature is that, unlike knowledge, understanding comes in degrees: one may understand something more or less well, and one's understanding may be better than another's. We introduce a comparative epistemic logic of understanding with graded modalities $\UO_i^\tau$ and a comparative connective $(i\succ j)\phi$ for ``$i$ understands why $\phi$ better than $j$''. Semantically, we enrich multi-agent epistemic models with agent-indexed graded explanation structures and a justification-style term algebra. This yields a unified framework for representing minimal, ordinary, more demanding, and ideal understanding, together with comparisons between agents with respect to the same formula at issue. We distinguish a finitary bounded-level calculus from an infinitary full-language companion system. We establish soundness and strong completeness, and show that each fixed finite-level fragment is decidable.
		\end{abstract}
		
		\section{Introduction}
			\emph{Why does the Earth orbit the Sun?}
			This question can be simple enough for a classroom.
			A child may be credited with understanding why by giving a basic gravity story
			(e.g., ``the Sun's gravity keeps the Earth going around''); in another context, say in an
			undergraduate astrophysics seminar, that attribution may be withdrawn once stricter explanatory
			standards are imposed.\footnote{We borrow this example from \cite{Egler2024-WUC}.}
			The withdrawal does not mean ``no understanding at all'': it means the required explanatory standard has shifted.
			The same pattern repeats at higher levels.
			An undergraduate physics student may count as understanding in a seminar, yet not in a panel of astrophysics experts discussing nearby phenomena.
			Again, the point is not that the student's understanding evaporates.
			The point is that what it takes to qualify as understanding in that context is more epistemically demanding. This is exactly the phenomenon of understanding coming in degrees.
			As noted in \cite{sliwa2015iv}, one prominent reason for distinguishing understanding from knowledge is that understanding is widely taken to admit degrees, whereas knowledge is typically not treated in this way.

	It already reveals the core logical issue of the present work. Besides, the same example also displays a \emph{comparative} dimension. It is natural to say that some people understand why something is the case better than others.
	Two agents may both explain why the Earth orbits the Sun, yet one explanation can be deeper,
	more integrated, or more counterfactually robust.
	We therefore need to express not only whether an agent understands why a proposition holds, but also whether one agent understands it better than another. Modern epistemic logic has rich tools for knowledge and belief, but much less machinery for these two features taken together: degree-sensitive understanding and comparative understanding.
	
		This logical gap mirrors a deeper epistemological distinction: standard philosophical cases separate knowing from understanding.
	A child may know via testimony that faulty wiring caused a fire while lacking the ability to
	explain the relevant mechanism.
	A scientist may identify oxygen as a decisive factor in a reaction yet still lack 
	understanding of why that dependence holds
	\cite{pritchard2014knowledge,lawler2019understanding}.
		These stories support, in particular, the contrast between knowing why and understanding why. They suggest not only a non-identity claim: understanding is often taken to require explanatory grasp beyond merely possessing true information, and is therefore a richer epistemic achievement than bare knowing.
	Given this distinction and its epistemic significance, understanding is now a central topic
	in contemporary epistemology and philosophy of science rather than a marginal one
	\cite{khalifa2013role,Baumberger2017-BAUWIU,grimm201610}.

	The issue is equally salient in AI-facing debates. In the wake of large language models,
	disputes about whether, and in what sense, AI genuinely understands have moved to the center
	of discussion \cite{mitchell2024debate,beckmann2025mechanistic}. Yet the pressure is not new:
	well before the recent LLM wave, authors had already noted that ``understanding'' is often
	invoked operationally without a stable theoretical account \cite{thorisson2016about}.
	If formal epistemology is to contribute here, it must represent not only knowledge but also
comparative understanding.

	There is also a historical reason to revisit the topic.
	Understanding was not alien to earlier logical traditions: medieval epistemic logic
	treated it as an epistemic mode in its own right \cite{boh1993epistemic,boh2000four}.
	As Boh emphasizes, ``the epistemic modality of understanding seems to be treated as even
	more basic than knowing or believing'' \cite[p.~100]{boh1993epistemic}.
	Our aim is to recover understanding as a formally disciplined notion within contemporary
	epistemic logic, while keeping close contact with current philosophical discussions.
	Formally, we build on \cite{yu2024understanding,xu2016logic} by moving to a graded
	and comparative setting that integrates an epistemic base, a justification-style
	explanation algebra, and an explicit comparative connective within one framework.

	\subsection{Background and Related Work}

	Philosophers distinguish multiple uses of ``understanding'':
	understanding-that, understanding-wh, and objectual understanding
	\cite{gordon2012there,baumberger2014types}.
Among these, understanding-wh (especially understanding-why) is often treated as the central case \cite{Khalifa2017-KHAUEA}.	We follow this line and focus on understanding-why.

	
	A central philosophical thesis is that understanding is explanation-involving.
	Understanding why is widely characterized as ``\textit{explanatory understanding}'',\footnote{
		For example, see \cite{Baumberger2017-BAUWIU,Khalifa2017-KHAUEA} and the
		bibliographies therein. We do not treat the ``why'' in ``understanding why'' as restrictive: some cases are
		more naturally phrased as \textit{understanding how} \cite{Khalifa2017-KHAUEA}.
		For example, one may say ``understanding how the dinosaurs went extinct'' rather than
		``understanding why they went extinct,'' without intending any substantial difference.
	}
	which already signals this connection.
	Wilkenfeld argues that explanations are the kinds of things that bring about understanding
	\cite{Wilkenfeld2014-WILFEA-4}, and Strevens puts the point succinctly: ``No understanding
	without explanation'' \cite{Strevens2013-STRNUW}.
	At the same time, theories of explanation differ sharply over what counts as explanans and
	explanatory support \cite{Hempel1948-HEMSIT-3,Hempel1965-HEMAOS,Salmon1985-SALSEA-9,sep-scientific-explanation}. For a logic with broad applicability, this suggests a methodological stance: explanatory structure should be
	represented abstractly enough to avoid commitment to any one substantive theory, but concretely enough for logical analysis. 
	A nearby line treats understanding as involving compression: understanding is not a matter of storing a long list of disconnected facts, but of having a compact representation of relevant structure that can be used to recover relevant
	information about the target phenomenon \cite{Wilkenfeld2018-UAC,Carbonell2025-CGCU}.
	This point fits the present project well: useful compression must still retain enough information to support explanation, and the formal framework below will distinguish more specific explanations from coarser explanatory resources that may provide weaker support than their more specific counterparts.

	
		Recent epistemic logic has extensively studied non-standard knowledge operators
		(what/how/why, etc.); see \cite{wang2018beyond}.
		A key predecessor for our setting is the logic of knowing why
		\cite{xu2016logic}, which combines epistemic accessibility with explanation terms.
		Intuitively, agent $i$ knows why $p$ when $i$ knows that $p$ and has a uniform explanation
		witness that works across all $i$-accessible worlds.
		In this sense, knowing-why has an existential explanation profile, which can be rendered schematically
		as $\exists t\,\K_i(t:p)$ in a justification-style metalanguage. 
		Philosophically inspired by \cite{lawler2019understanding} and technically by \cite{xu2016logic},
		Wei \cite{yu2024understanding} models understanding-why via higher-order explanation,
		schematically $\exists t_1\exists t_2\,\K_i\!\bigl(t_2:(t_1:\phi)\bigr)$, where
		$t_1:\phi$ means that $t_1$ is an explanation for $\phi$, and $t_2:(t_1:\phi)$ means that
		$t_2$ is a \textit{higher-order} explanation of ``$t_1$ explains $\phi$''.
		This logic thereby embodies the philosophical idea that understanding why requires at least two explanations at different levels, beyond
		what standard knowing-why requires.

	The present paper builds directly on this line. The predecessor captures the key insight that understanding requires higher-order explanatory support, but its formal architecture remains too coarse for a full hierarchy of explanatory depth
	and lacks an explicit comparative connective with a unified meta-theory. We therefore replace the packed operator with a level-indexed family, add a
	comparative connective, and use graded explanations in the semantics. Intuitively, higher grades mark explanations as eligible to support higher-level, and hence more demanding, understanding claims.

This work is best read as a synthesis of epistemic, comparative, and
justification-style ideas.
At its base, the framework remains an epistemic logic in the strict technical sense:
its semantics is based on multi-agent epistemic models, and modal interaction principles are developed inside that setting.
In this respect, the project continues the ``beyond knowing that'' line in
contemporary epistemic logic, but shifts the focus to understanding-why and comparative understanding.

The language is also genuinely comparative, because it contains an explicit
connective $(i\succ j)\phi$ stating that agent $i$ understands why $\phi$ better than $j$.
Comparative modal work such as \cite{ditmarsch2009knowingmore} already includes a comparative operator $i\succeq j$ to express that, locally at a given world, all formulas known by agent $i$ are also known by agent $j$.
In that framework, the global axiom scheme $\K_i\psi\to\K_j\psi$ is replaced by a local inference-rule treatment.
Our approach follows this local-comparative spirit, but with a different semantic source.
The comparative clause is not a primitive relation-inclusion postulate; it is induced by
explanation-sensitive understanding levels and evaluated relative to the formula at issue, with an epistemic side condition.
This design has substantial technical consequences, developed in detail later.

		The semantic machinery also borrows core ingredients from justification logic:
		explanation terms, algebraic term operations, and admissibility constraints
		\cite{fitting2005logic,artemov2008logic,artemov2019justification}. Earlier logics \cite{xu2016logic,yu2024understanding} omit the sum operator $+$, understandably: a disjunctive or choice term can be too indeterminate to serve as a determinate ``reason-why'' witness in ordinary knowledge-why ascriptions. 	\footnote{For example, in a simple finite case, the agent may consider several possible
			causes of a forest fire: lightning, arson, an unattended campfire, or a
			power-line fault. If each alternative has its own explanation term, repeated use
			of $+$ can combine these terms into one disjunctive witness. This may provide a uniform witness across these alternatives, but the combined
			term no longer specifies which cause explains the fire at the actual world. See Xu et al.\
			\cite{xu2016logic}.} The present framework nevertheless includes the full term algebra
			$(\cdot,+,!,c)$. For present purposes, $+$ has a natural explanation-theoretic reading: $t+s$ is available as an explanation of $\phi$ whenever either $t$ or $s$ is available. The compression perspective mentioned above explains why this operation should be treated with care. It does not require us to exclude $+$; rather, it shows that a sum term preserves only the fact that some explanation is available, while losing the information about which explanation supplies the support. For this reason, $+$ is included, but it is strictly grade-downgrading. Thus the system is justification-style, but not a standard justification logic: the object language has no formulas of the form $t:\phi$. Explanation terms enter only semantically, as graded witnesses governed by the agent-indexed
			functions $\Ex_i$ and the epistemic accessibility relations used for knowledge and comparison.

		\subsection{Framework and Contributions}

		We propose a comparative epistemic logic of understanding with two operator families:
	$\UO_i^\tau\phi$ for level-indexed understanding-why and $(i\succ j)\phi$ for comparative understanding with respect to a formula. 
			At the conceptual level, we take finite indices $\tau$ to range over $\mathbb N^{+}=\{1,2,3,\ldots\}$. This already lets the language capture the spectrum emphasized in
		\cite{Khalifa2017-KHAUEA}: minimal understanding $<$ everyday understanding
		$<$ typical scientist's understanding $<$ ideal understanding. Schematically, these may be represented by
		$\UO_i^1\phi,\UO_i^2\phi,\UO_i^m\phi,\UO_i^\omega\phi$.
		Level $1$ is intended to model a minimal explanatory foothold
		(a necessary but not yet typical notion of understanding),
		while level $2$ is intended to capture ordinary or everyday explanatory understanding
		in a higher-order sense \cite{lawler2019understanding,yu2024understanding}.
		Some finite level $m>2$ can then be used to represent typical scientific understanding.
		Higher levels correspond to explanations that are more systematically organized
		and more deeply integrated into scientific understanding.
		This level-sensitive hierarchy fits Khalifa's spectrum picture, in which ideal understanding
		is a limit notion \cite{Khalifa2017-KHAUEA}. 	Formally, we capture that limit by requiring every finite level:
		$\bigwedge_{n\in \mathbb N^{+}}\UO_i^n\phi$.
		
The numerical levels are a formal idealization. Finite indices on the understanding operators represent an ordered family of explanatory standards
fixed in a model; a larger index marks a more demanding standard for the same formula. Depending on the application, such standards may concern specificity, systematic integration, or counterfactual robustness. The logic does not decide which substantive features make one standard stronger than another in every
application; it only provides a framework for representing such standards once fixed. Meeting a stronger standard entails meeting weaker standards for the same formula, not possessing every particular explanation that might satisfy a weaker
standard. 

	Our main technical contributions are fourfold. First, we give a graded semantics for level-indexed and comparative understanding, with comparisons always relative to the formula at issue. Second, for each fixed finite understanding-level domain, we introduce a finitary calculus and establish soundness, strong completeness, and decidability. Third, for the full language, we move to an infinitary extension of the finitary calculus, reflecting the ideal-understanding level, and prove soundness and strong completeness. Fourth, we isolate the boundary between the bounded and full languages: in each fixed finite-level fragment the comparative connective is eliminable, while in the full language it is not; moreover, the full language is non-compact, so no sound finitary proof system can be strongly complete for it.

		
		The rest  follows this architecture.
		Section~\ref{setting} introduces syntax and semantics.
		Section~\ref{III} presents the calculi and soundness.
		Section~\ref{sec:completeness} proves bounded and full completeness and gives a
		bounded decidability result.
		Section~\ref{sec:conclusion} concludes.
	
	\section{Syntax and Semantics}\label{setting}
		
			\begin{definition}
				Given nonempty countable sets $P$ of proposition letters and $I$ of agents, and a
				nonempty understanding-level domain $L\subseteq \mathbb N^{+}\cup\{\omega\}$, where
				$\mathbb N^{+}:=\{1,2,3,\dots\}$, the comparative epistemic language of understanding
				$\ELBU(L)$ is defined by
				(where $p\in P$, $i,j\in I$, $\tau\in L$):
			\[
			\phi::=~p \mid \neg\phi\mid (\phi\wedge\phi)\mid \K_i\phi
			\mid \UO_i^\tau\phi \mid (i\succ j)\phi.
			\]
		We use only these two instances in what follows:
			\[
			\ELBU:=\ELBU(\mathbb N^{+}\cup\{\omega\}),\qquad
			\ELBU^\kappa:=\ELBU(\{1,\dots,\kappa\})\ \ (\kappa\geqslant2).
			\]
		\end{definition}
		
		Here $\tau$ is an \emph{understanding-level index}. As explained above, finite
		indices on $\UO_i^\tau$ are intended to distinguish explanatory standards of
		different strength. Thus $\UO_i^n\phi$ says that agent $i$ meets the $n$th such standard for understanding why $\phi$. Note that we set $\Kyy_i\varphi $ in \cite{xu2016logic,yu2024understanding} as $ \UO_i^1\varphi$ now, which can be read as minimal understanding.\footnote{It may be tempting to set $\K_i\varphi :=\UO_i^0\varphi $, i.e., knowledge is just level-0 understanding. However, our philosophical point is that understanding is more than knowledge, which does not require that $\K_i\varphi$ itself is already a very minimal form of understanding.}  The formula $\UO_i^\omega\phi$ expresses ideal understanding and is interpreted as ``for all finite levels $n$, $\UO_i^n\phi$''. By abuse of notation, we write \(\UO\) for the family of modalities \(\UO_i^\tau\). The comparative understanding formula $(i\succ j)\phi$ indicates that agent $i$ understands why $\phi$ better than agent $j$ does.
		
		
		Informally, the bounded language
		$\ELBU^\kappa$ keeps comparatives but restricts all understanding-level indices
		to $\{1,\dots,\kappa\}$, so ideal and unbounded distinctions are not expressible. 
		Technically, $\ELBU^\kappa$ is the base language for both the finitary calculus
		and the bounded decidability analysis.

		
		
		
		
		We accept the view in \cite{xu2016logic}  that although something is a tautology, one may still lack minimal understanding (knowledge why) of that tautology.
		A special set of ``self-evident'' tautologies $\Lambda$ is introduced,  which the agent is assumed to minimally understand. For example, we can let all the instances  of $\phi\wedge\psi\rightarrow\phi$ and $\phi\wedge\psi\rightarrow\psi$ be $\Lambda$. Such simple choices will behave well in the bounded decidability
		argument below. At present, we do not suppose any necessitation rule for $\UO$ in general.

		\begin{definition}\label{cu}
			Fix a level domain $L$ and the associated language $\ELBU(L)$.
			A graded explanatory epistemic $\ELBU(L)$-model $\M$ is a tuple
			$(W,\{R_i\mid i\in I\},V,E,\mathrm{gr},\{\Ex_i\mid i\in I\})$ where
			$(W,\{R_i\mid i\in I\},V)$ is a standard multi-agent epistemic model, i.e. for each $i\in I$, $R_i$ is an equivalence relation on $W$, and:
			\begin{itemize}
				\item $E$ is a nonempty set of explanations, closed under $\cdot$, $+$, and $!$, and containing the constant $c$.			

				\item $\gr:E\to \mathbb N$ is a grade map. For each $n\in\mathbb N$,
				let $E_n:=\{t\in E\mid \mathrm{gr}(t)\geqslant n\}$.
				The grade map satisfies:
				\begin{itemize}
					\item $\gr(t\cdot s)\geqslant \min\{\gr(t),\gr(s)\}$;
					\item 
					if $\min\{\gr(t),\gr(s)\}=0$ then $\gr(t+s)=0$;
					
					if $\min\{\gr(t),\gr(s)\}\geqslant 1$ then $\gr(t+s)<\min\{\gr(t),\gr(s)\}$.
					\item $\gr(!t)=0$ if $\gr(t)=0$, and $\gr(!t)=2$ if $\gr(t)\geqslant 1$;
					\item $\gr(c)=1$.
				\end{itemize}
					
					


				\item $\{\Ex_i\mid i\in I\}$ is a family of admissible explanation functions, each
				$\Ex_i:E\times \ELBU(L)\to 2^W$ satisfies:
				\begin{description}
					\item[\textit{Explanation Application}] $\Ex_i(t,\varphi\to\psi)\cap \Ex_i(s,\varphi)\subseteq \Ex_i(t\cdot s,\psi)$.
					\item[\textit{Explanation Sum}] $\Ex_i(t,\varphi)\cup \Ex_i(s,\varphi)\subseteq \Ex_i(t+s,\varphi)$.
					\item[\textit{Constant Specification}] If $\phi\in\Lambda$, then $\Ex_i(c,\phi)=W$.
					\item[\textit{Epistemic Introspection}] For all $t\in E$, $\Ex_i(t,\K_i\varphi)\subseteq \Ex_i(!t,\K_i\varphi)$.
				\end{description}

			\end{itemize}	
		\end{definition}
		
	The term operations $\cdot$, $+$, $!$, and the constant $c$ are standard in justification logic. Application $\cdot$ combines explanations, $+$ has the	usual disjunctive or choice reading, $!$ represents positive introspection, and
	the constant $c$ serves as a self-evident explanation for formulas in	$\Lambda$.
	Unlike the models in \cite{xu2016logic,yu2024understanding}, every explanation term here also carries a grade.

There is another generalization. In \cite{xu2016logic,yu2024understanding}, a single admissible explanation function, not indexed by agents, records for each term $t$ and formula $\phi$ the worlds at which $t$ explains $\phi$. Here we replace that single function with an agent-indexed family $\{\Ex_i\}_{i\in I}$. Thus $w\in\Ex_i(t,\phi)$ means that, at $w$, $t$ is
available as an explanation of $\phi$ for agent $i$. 
 In this respect, the present model is more general: explanatory support can vary with the subject whose understanding is at issue.

	The grade constraint on $+$ works together with the Sum condition on $\Ex_i$. The Sum condition preserves availability: if either $t$ or $s$ explains $\phi$ for agent $i$ at a world, then $t+s$ does too. But the sum term $t+s$ is less specific than either input, since it does not by itself specify whether $t$ or $s$ is the available explanation at that world. For this reason, $+$ is strictly
	grade-lowering whenever the input minimum is positive. Grade-$0$ terms are allowed as closure-generated explanation terms. They belong to $E_0$, but not to any positive-grade class $E_n=\{t\in E\mid \gr(t)\geqslant n\}$ with $n\geqslant1$. As the truth conditions below make explicit, grade-$0$ terms cannot support
	positive-level understanding claims.
	
	The first three conditions on $\Ex_i$ are standard. The fourth condition, $\Ex_i(t,\K_i\varphi)\subseteq \Ex_i(!t,\K_i\varphi)$,
	requires a separate reading. Its source is the positive-introspection principle from
	justification logic. In standard justification logic,
	$t:\varphi\to !t:(t:\varphi)$ says that if $t$ is a justification for
	$\varphi$, then $!t$ is a justification for the claim that $t$ justifies
	$\varphi$ \cite{fitting2004logic}. Intuitively, $!t$ records a reflective
	confirmation of the justificatory status of $t$.
	
	The introspection condition uses this idea only for knowledge claims about the same agent. Suppose $p_{\mathrm{orb}}$ says that the Earth orbits the Sun. A request to explain why agent $i$ knows $p_{\mathrm{orb}}$ is normally a request for $i$'s reasons or evidence for believing $p_{\mathrm{orb}}$, not a	request to explain why the Earth orbits the Sun. This is in line with the view
	that requests to justify knowledge claims commonly ask for the agent's reasons or evidence \cite{mckinnon2012you}. In such cases, these reasons can play a justification-like explanatory role: an explanation of $\K_i\varphi$	gives the agent's reasons or evidence for believing $\varphi$.

	The condition above says that whenever, at a world $w$, $t$ is available to
	agent $i$ as an explanation of $\K_i\varphi$, the reflected term $!t$ is also
	available to $i$ at $w$ as an explanation of the same knowledge claim. The
	point is not that $!t$ adds another explanation of $\varphi$ itself. Rather,
	$!t$ makes explicit that $i$ can cite $t$ as the support for knowing
	$\varphi$.
	
	This should be read together with the grading rule for $!$: if $\gr(t)=0$ then
	$\gr(!t)=0$, while if $\gr(t)\geqslant1$ then $\gr(!t)=2$. Level $1$ is minimal
 understanding, corresponding to knowing why; applied to $\K_i\phi$, it means that agent $i$ has an explanation of why she knows $\phi$. The term $!t$ represents reflecting on that reason as her reason, so for own-knowledge claims level $2$ marks this reflective grasp. Thus the semantic condition validates only the restricted introspection principle $\UO_i^1\K_i\phi\to\UO_i^2\K_i\phi$: it is not a general mechanism for producing ever higher levels, does not validate $\UO_i^n\K_i\phi\to\UO_i^{n+1}\K_i\phi$, and does not apply to $\UO_i^1\K_j\phi\to\UO_i^2\K_j\phi$ for $j\neq i$ or to non-epistemic formulas.

		\begin{definition}[Truth conditions]
			Fix a level domain $L$, a $\ELBU(L)$-model $\M$, and let
			$L_{\mathrm{fin}}:=L\cap\mathbb N^{+}$.
			The satisfaction relation for $\ELBU(L)$ is defined inductively as follows
			(for $p\in P$, $i,j\in I$, and $n\in L_{\mathrm{fin}}$):
			\begin{center}
				\begin{tabular}{|lrl|}	\hline
					$\M,w\vDash p$ &$\Leftrightarrow$ &$w\in V(p)$\\
					$\M,w\vDash \neg \phi$ &$\Leftrightarrow$ & $\M,w\not\vDash \phi$ \\
					$\M,w\vDash \phi\wedge \psi$ &$\Leftrightarrow$ & $\M,w\vDash \phi$ and $\M,w\vDash \psi$\\
					$\M,w\vDash \K_i \phi$ &$\Leftrightarrow$ & $\M,v\vDash \phi$ for all $v$ such that $wR_iv$ \\
					$\M,w\vDash\UO_i^n \phi$ &$\Leftrightarrow$ &
					(1) there exists $t\in E_n$ such that \\
					& & \quad for all $v\in W$ with $wR_iv$, $v\in \Ex_i(t,\phi)$,\\
					& & (2) for all $v\in W$ with $wR_iv$, $\M,v\vDash \phi$\\
					$\M,w\vDash\UO_i^\omega \phi$ &$\Leftrightarrow$ &
					(if $\omega\in L$) for all $n\in L_{\mathrm{fin}}$, $\M,w\vDash\UO_i^n \phi$\\
					$\M,w\vDash(i\succ j)\phi$ &$\Leftrightarrow$ &
					(1) $\deg_i(w,\phi)>\deg_j(w,\phi)$, and\\
					& & (2) $\M,w\vDash \K_j\phi$\\
					\hline
				\end{tabular}
			\end{center}
			where the \emph{understanding degree function} is $\deg_i(w,\varphi):=\sup\bigl(\{n\in L_{\mathrm{fin}}\mid \M,w\vDash \UO_i^n\varphi\}\cup\{0\}\bigr)$.
		\end{definition}

Thus $\deg_i(w,\varphi)$ is not a primitive measure, but the supremum of the finite levels at which $i$ understands why $\varphi$ at $w$. The existential
quantifier in $\UO_i^n\phi$ is a witness condition: a level-$n$ claim requires one and the same explanation term of grade at least $n$ to be available at every $i$-accessible world.


If $\omega\notin L$, as in the bounded languages $\ELBU^\kappa$, then $\UO_i^\omega\phi$ is not a well-formed formula, so the $\UO^\omega$ clause is absent. In the full language $\ELBU$, where
$L=\mathbb N^{+}\cup\{\omega\}$, $\UO_i^\omega\phi$ expresses ideal understanding as the limit case requiring $\UO_i^n\phi$ for every $n\in\mathbb N^{+}$. Equivalently, $\M,w\vDash \UO_i^\omega\varphi$ iff $\deg_i(w,\varphi)=\omega$.

Although $(i\succ j)\varphi$ is primitive in the syntax, its semantic value is
fixed by two independently defined notions: the induced degrees $\deg_i,\deg_j$
and the epistemic condition $\K_j\varphi$. The comparison is formula-relative:
it compares agents only with respect to the same $\varphi$, not globally. The
truth condition states $\K_j\varphi$ explicitly because
$\deg_i(w,\varphi)>\deg_j(w,\varphi)\geqslant0$ already implies
$\M,w\vDash\UO_i^1\varphi$, hence $\M,w\vDash\K_i\varphi$.

$(i\succ j)\varphi$ covers both cases where both agents understand $\varphi$ but at different levels, and cases where $i$ has minimal understanding while $j$ merely knows $\varphi$. The latter remains a comparison within a shared epistemic issue: having an explanation already supports the ordinary judgment ``I understand it better''.

Accordingly, for finite $m$ with $m,m+1\in L_{\mathrm{fin}}$, the conjunction
$\UO_i^m\varphi\wedge\neg\UO_i^{m+1}\varphi$ isolates the exact finite
understanding degree $m$ of agent $i$ with respect to $\varphi$.

Returning to the opening example, let $p_{\mathrm{orb}}\in P$ say that the Earth
orbits the Sun, and let $i_C,i_S,i_R\in I$ denote the child, the student, and the
researcher. In one natural formalization, a world $w$ may satisfy
$1\leqslant\deg_{i_C}(w,p_{\mathrm{orb}})<\deg_{i_S}(w,p_{\mathrm{orb}})<\deg_{i_R}(w,p_{\mathrm{orb}})$.
Then $\M,w\vDash(i_S\succ i_C)p_{\mathrm{orb}}$ and
$\M,w\vDash(i_R\succ i_S)p_{\mathrm{orb}}$. If a bystander $i_B$ merely knows
that $p_{\mathrm{orb}}$ while $\deg_{i_B}(w,p_{\mathrm{orb}})=0$, then
$\M,w\vDash(i_C\succ i_B)p_{\mathrm{orb}}$ holds.

		\begin{remark}
			$(i\succ j)\varphi$ becomes eliminable in the comparative-free fragment of $\ELBU^\kappa$.
			That is, for every model $\la\M,w\ra$ over level domain $L=\{1,\dots,\kappa\}$,
			\[
			\M,w\vDash (i\succ j)\varphi
			\Longleftrightarrow
			\M,w\vDash  \K_j\varphi\wedge
			\bigvee_{m=1}^{\kappa}\bigl(\UO_i^{m}\varphi \wedge \neg \UO_j^{m}\varphi\bigr).
			\]
		\end{remark}

		\begin{proposition}\label{prop:succ-not-def}
			In $\ELBU$, the $(i\succ j)\varphi$ is in general not eliminable by any comparative-free formula.
		\end{proposition}
		
			\begin{proof}
				It suffices to show non-eliminability for one instance, say $(i\succ j)p$.
				Suppose $\chi$ is a comparative-free formula. Since $\chi$ contains only finitely many finite understanding-level indices, let $n$ be greater
				than all those indices.
					Take two pointed models $(\M,w)$ and $(\M',w')$, each with just one world,
					with the same valuation, in particular with $p$ true. All accessibility
					relations are reflexive, so $\K_jp$ holds at both points. Specify the explanation functions by putting
				$w\in\Ex_i(t_i,p)$, $w\in\Ex_j(t_j,p)$,
				$w'\in\Ex'_i(t_i,p)$, and $w'\in\Ex'_j(t_j,p)$, and by adding only what is
				required by the admissibility clauses. Write $\gr$ and $\gr'$ for the grade maps of
					$\M$ and $\M'$, respectively, and let the only relevant difference be the grades:
					$\gr(t_i)=n$, $\gr'(t_i)=n+1$, and $\gr(t_j)=\gr'(t_j)=n$. Choose the remaining
					grades so that all generated terms have bounded grade; hence every $\UO^\omega$-formula is false
					at both points. Then
				$\deg_i(w,p)=n$, $\deg_i(w',p)=n+1$, and
					$\deg_j(w,p)=\deg_j(w',p)=n$. The two points satisfy the same comparative-free formulas whose finite indices are
					below $n$: this is proved by induction on formulas, with the $\UO_k^m$ case using
					the same available witnesses, if any, for every $m<n$, and with the
					$\UO^\omega$ case handled by the preceding boundedness observation. Hence
				$\M,w\vDash\chi \iff \M',w'\vDash\chi$.	But $\M,w\not\vDash (i\succ j)p$, and
				$\M',w'\vDash(i\succ j)p$.
			\end{proof}
		
		Below we show that 	the full language $\ELBU$ is non-compact.
		\begin{proposition}\label{prop:noncompact-full}
			Fix distinct agents $i\neq j$ and an atom $p$. Let
			$\Sigma_{\succ}:=\{(i\succ j)p\}\ \cup\ \{\UO_j^n p\mid n\in\mathbb N^{+}\}$.	Then every finite subset of $\Sigma_{\succ}$ is satisfiable over $\ELBU$-models,
			but $\Sigma_{\succ}$ itself is unsatisfiable.
		\end{proposition}
		\begin{proof}[Proof sketch]
			For any finite $\Delta\subseteq\Sigma_{\succ}$, let $m$ be the largest $n$ such that $\UO_j^n p\in\Delta$. Then a one-state model with $p$ true, reflexive accessibility, and	$\deg_j(w,p)=m< m+1=\deg_i(w,p)$ satisfies $\Delta$ and $(i\succ j)p$.	
			
			If $\M,w\vDash\UO_j^n p$ for all $n\geqslant1$, then
			$\deg_j(w,p)=\sup\bigl(\{n\mid \M,w\vDash\UO_j^n p\}\cup\{0\}\bigr)=\omega$. So $\M,w\vDash(i\succ j)p$ is impossible, since it requires $\deg_i(w,p)>\deg_j(w,p)=\omega$.
			Hence $\Sigma_{\succ}$  is unsatisfiable.
		\end{proof}
	Therefore, 	no finitary proof system over $\ELBU$ can be both sound and strongly complete.
		

		
	As in \cite{xu2016logic,yu2024understanding}, explanation factivity, defined below, is not built into the model definition.
		\begin{definition}
			Fix a level domain $L$ and a $\ELBU(L)$-model $\M$.
			We say that $\M$ has \emph{explanation factivity} if whenever
			$w\in \Ex_i(t,\phi)$ for some $i\in I$ and $t\in E$, then $\M,w\vDash \phi$.
		\end{definition}
		
		Given a $\ELBU(L)$-model
		$\M=(W,\{R_i\mid i\in I\},V,E,\mathrm{gr},\{\Ex_i\mid i\in I\})$, define its factive companion
		$\M^F=(W,\{R_i\mid i\in I\},V,E,\mathrm{gr},\{\Ex_i^F\mid i\in I\})$ by $\Ex_i^{F}(t,\phi)=\Ex_i(t,\phi)\setminus\{w\mid \M,w\not\vDash \phi\}$.
		By direct checking, $\M^F$ is again a $\ELBU(L)$-model.
		The proposition below asserts that $\ELBU(L)$-truth is invariant under this factive transformation.
		The proof is routine and omitted for reasons of space.
			\begin{proposition}
			For any $\ELBU(L)$-formula $\phi$ and any $w\in W$, $\M,w\vDash \phi \Leftrightarrow \M^{F},w\vDash \phi$.
		\end{proposition}

		
		%

		\section{Axiomatization}\label{III}
	We present two related systems over two aligned languages:
		the bounded finitary calculus $\SCUzero$ for $\ELBU^\kappa$, and
		the infinitary system $\SCUi$ for $\ELBU$.
		
		\begin{definition}[Finitary calculus]\label{def:scuzero-system}
			Fix $\kappa\geqslant2$.
			Let $\SCUzero$ be the finitary Hilbert system over $\ELBU^\kappa$ whose axiom schemas are
			(where $i,j,k\in I$, and understanding-level indices range over $L=\{1,\dots,\kappa\}$):
			\begin{alignat*}{2}
				(\TAUT)&\quad \text{Propositional tautologies} &&\\
				(\Kg)&\quad \K_i(\phi\to\psi)\to(\K_i\phi\to\K_i\psi) &&\\
				(\T)&\quad \K_i\phi\to\phi &&\\
				(\4)&\quad \K_i\phi\to\K_i\K_i\phi &&\\
				(\5)&\quad \neg\K_i\phi\to\K_i\neg\K_i\phi &&\\
				(\Kgg)&\quad \UO_i^n(\phi\to\psi)\to(\UO_i^n\phi\to\UO_i^n\psi) && \quad (n\in L)\\
				(\IYB)&\quad \UO_i^n\phi\to\UO_i^m\phi && \quad (n,m\in L,\ n>m)\\
				(\IMP)&\quad \UO_i^n\phi\to\K_i\phi && \quad (n\in L)\\
				(\4^*)&\quad \UO_i^n\phi\to\K_i\UO_i^n\phi && \quad (n\in L)\\
				(\KYU)&\quad \UO_i^1\K_i\phi\to\UO_i^2\K_i\phi&&\\
				(\UYC0)&\quad \UO_i^{1}\phi\wedge\K_j\phi\wedge\neg \UO_j^{1}\phi\rightarrow(i\succ j)\phi &&\\
				(\UYC)&\quad \UO_i^{n}\phi\wedge\UO_j^{m}\phi\wedge\neg \UO_j^{m+1}\phi\rightarrow(i\succ j)\phi
				&& \quad (n,m,m+1\in L,\ n>m)\\
				(\CMPn)&\quad (i\succ j)\phi\wedge\UO_j^n\phi\rightarrow \UO_i^{n+1}\phi
				&& \quad (n,n+1\in L)\\
				(\CMP1)&\quad (i\succ j)\phi\rightarrow \UO_i^1\phi\wedge\K_j\phi &&\\
				(\CMP\kappa)&\quad (i\succ j)\phi\rightarrow \neg \UO_j^\kappa\phi. &&
			\end{alignat*}
			The inference rules are:
			\begin{center}
				\begin{tabular}{llllll}	
					$(\MP)$& Modus Ponens&	$(\Nn)$& $\vdash \phi /\ \vdash \K_i\phi$&$(\NE)$& $\phi\in \Lambda /\ \vdash\UO_i^1\phi$\\
					
				\end{tabular}
			\end{center}
		\end{definition}
			The axiom $(\KYU)$ captures a restricted introspective step: from
			level-$1$ understanding of why agent $i$ knows $\phi$, agent $i$ can move
			to level-$2$ understanding of that same knowledge claim by reflecting on
			her own reason for knowing $\phi$. It is the proof-theoretic counterpart of the semantic
			\textit{epistemic introspection} condition on $\Ex_i$.
			Axiom $(\CMP\kappa)$ records the upper bound of the finite language
			$\ELBU^\kappa$: if $i$ understands $\phi$ better than $j$, then $j$ cannot
			already satisfy the top available level $\kappa$.In the full system below, it is replaced by $(\CMP\omega)$, which plays the same role for ideal understanding.
		
		\begin{definition}[Full calculus]\label{def:scui-system}
			The system $\SCUi$ over $\ELBU$ is the $\omega$-companion of $\SCUzero$:
			it has the same finitary schemas/rules as Definition~\ref{def:scuzero-system}
			with all finite level indices ranging over $\mathbb N^{+}$, except that the finite-language upper-bound axiom $(\CMP\kappa)$ is replaced by
			\[
			(\CMP\omega)\quad (i\succ j)\phi\rightarrow \neg \UO_j^\omega\phi,
			\]
			and it additionally includes the following axiom schema:
			\[
			(\IYB_\omega^n)\quad \UO_i^\omega\phi\rightarrow \UO_i^n\phi\ \ (n\in\mathbb N^{+}),
			\]
			and the following infinitary rule schema:
			\[
			(\omega \mathrm{I})\quad
			\frac{
				\{\chi\to
				\K_{j_1}(\theta_1\to\cdots\K_{j_m}(\theta_m\to\UO_i^n\phi)\cdots)
				\mid n\in\mathbb N^{+}\}
			}{
				\chi\to
				\K_{j_1}(\theta_1\to\cdots\K_{j_m}(\theta_m\to\UO_i^\omega\phi)\cdots)
			}.
			\]
				Here $m\geqslant0$, $j_1,\dots,j_m\in I$, and
				$\theta_1,\dots,\theta_m\in\ELBU$. If $m=0$, the rule is read without
				any outer $\K$-operators, i.e.,
				$\{\chi\to\UO_i^n\phi\mid n\in\mathbb N^{+}\}/(\chi\to\UO_i^\omega\phi)$.
				If $m=1$ and $\theta_1=\top$, the rule includes the instance
				$\{\chi\to\K_j\UO_i^n\phi\mid n\in\mathbb N^{+}\}/
				(\chi\to\K_j\UO_i^\omega\phi)$.
				That is, under the same assumption $\chi$, if $j$ knows that $i$ has
				understanding of $\phi$ at every finite level,
				then $j$ also knows that $i$ has ideal understanding of $\phi$.
				Thus $m>0$ makes the same limit step available
				inside epistemic implication contexts.
		\end{definition}
		

	Since $\omega$-introduction is an infinitary rule with countably many premises,
	we represent derivability of system $\SCUi$ by well-founded proof trees, following the treatment of infinitary formal proofs in
	\cite{karp1964}.

			\begin{definition}[$\SCUi$-derivability]\label{def:scui-proof}
				A \emph{proof tree from $\Sigma$} is a well-founded tree whose nodes are labeled by formulas.
				Each node is either an assumption leaf (label in $\Sigma$), an axiom leaf
				(instance of a $\SCUi$-axiom), or is obtained by one of the following rules:
				\begin{itemize}
					\item $(\MP)$: from children labeled $\psi\to\chi$ and $\psi$, conclude $\chi$;
					\item $(\NE)$: conclude $\UO_i^1\psi$ for $\psi\in\Lambda$;
					\item $(\Nn)$: from one \emph{closed} child labeled $\psi$, conclude $\K_i\psi$;
					\item $(\omega \mathrm{I})$: for $m\in\mathbb N$,  from children
					labeled
					$\chi\to \K_{j_1}(\theta_1\to\cdots\K_{j_m}(\theta_m\to\UO_i^n\psi)\cdots)$,
					one for each $n\in\mathbb N^{+}$, conclude the corresponding formula with
					$\UO_i^\omega\psi$ in place of $\UO_i^n\psi$; 
					\end{itemize}
			where ``closed'' means that the corresponding subtree has no assumption leaves.
			We write $\Sigma\vdash_{\SCUi}\phi$ iff there exists such a proof tree with root label $\phi$.
			A set $\Sigma$ is \emph{$\SCUi$-consistent} iff $\Sigma\nvdash_{\SCUi}\bot$.
		\end{definition}
		
		\begin{lemma}\label{lem:scui-cut}
			If $\Sigma\vdash_{\SCUi}\psi$ and $\Sigma\cup\{\psi\}\vdash_{\SCUi}\chi$, then
			$\Sigma\vdash_{\SCUi}\chi$.
		\end{lemma}
			\begin{proof}
				Graft a fresh copy of a proof tree of $\psi$ from $\Sigma$ onto each assumption leaf
				labeled $\psi$ in a proof tree of $\chi$ from $\Sigma\cup\{\psi\}$, leaving all other nodes unchanged and preserving the original rule applications. The result is a proof tree from $\Sigma$. The only point requiring checking is
				the requirement in $(\Nn)$ that its premise child be closed: this is preserved because a closed child subtree contains no assumption leaves, so no grafting takes place inside it. Well-foundedness is also preserved: any descending chain either stays in the original tree or eventually enters a single grafted copy, both of which are well-founded.
			\end{proof}

	\begin{lemma}\label{lem:scui-neg-cons}
		If $\Sigma\nvdash_{\SCUi}\varphi$, then
		$\Sigma\cup\{\neg\varphi\}$ is $\SCUi$-consistent.
	\end{lemma}
	\begin{proof}
		Suppose that
		$\Sigma\cup\{\neg\varphi\}\vdash_{\SCUi}\bot$. By induction on proof trees, we first obtain:
		if $\Sigma\cup\{\alpha\}\vdash_{\SCUi}\beta$, then
		$\Sigma\vdash_{\SCUi}\alpha\to\beta$.
	The finitary cases are standard. For
	an $\omega$-introduction instance, write its premises as $\chi\to B_n$
	$(n\in\mathbb N^{+})$ and its conclusion as $\chi\to B_\omega$, where
	$B_\omega$ is obtained from $B_n$ by replacing the displayed
	$\UO_i^n\psi$ with $\UO_i^\omega\psi$. By the induction hypothesis,
		$\Sigma\vdash_{\SCUi}\alpha\to(\chi\to B_n)$ for all $n$. Hence
		$\Sigma\vdash_{\SCUi}(\alpha\wedge\chi)\to B_n$ for all $n$, so
		$(\omega \mathrm{I})$ gives
		$\Sigma\vdash_{\SCUi}(\alpha\wedge\chi)\to B_\omega$, and therefore
		$\Sigma\vdash_{\SCUi}\alpha\to(\chi\to B_\omega)$. Then we get $\Sigma\vdash_{\SCUi}\neg\varphi\to\bot$, therefore $\Sigma\vdash_{\SCUi}\varphi$, contradicting the hypothesis. So $\Sigma\cup\{\neg\varphi\}$ is consistent.
	\end{proof}

		
	The next derived principles are useful later in the completeness proofs. They also show that the formula $(i\succ j)\phi$ behaves like a strict comparison: irreflexivity, transitivity, and asymmetry follow from the axioms connecting comparison with level-indexed understanding. For the full system, we also derive the ideal-level analogues of the epistemic principles for understanding formulas.

		\begin{proposition}\label{prop:derived-order}
		The following principles are provable in the indicated systems:
			\[
			\begin{array}{r@{\quad}l@{\qquad}r@{\quad}l}
			(\5^*) &
			\neg\UO^n_i\phi\rightarrow \K_i\neg\UO^n_i\phi
			&
			(\4^*_\omega) &
			\UO_i^\omega\phi\rightarrow \K_i\UO_i^\omega\phi
			\\[1mm]
			(\5^*_\omega) &
			\neg\UO_i^\omega\phi\rightarrow \K_i\neg\UO_i^\omega\phi
			&
			(\Irref) &
			\neg(i\succ i)\phi
			\\[1mm]
			(\Trans) &
			(i\succ j)\phi\wedge (j\succ k)\phi\rightarrow (i\succ k)\phi
			&
			(\Asym) &
			(i\succ j)\phi\rightarrow\neg (j\succ i)\phi
			\end{array}
			\]
			Here $(\5^*)$, $(\Irref)$, $(\Trans)$, and $(\Asym)$ are provable in both
			$\SCUzero$ and $\SCUi$; for $(\5^*)$, $n\in\{1,\dots,\kappa\}$ in
			$\SCUzero$ and $n\in\mathbb N^{+}$ in $\SCUi$.  $(\4^*_\omega)$ and $(\5^*_\omega)$ are provable in $\SCUi$.
		\end{proposition}
		\begin{proof}
			The derivation of $(\5^*)$ is routine from $(\T),(\5),(\4^*)$ together with classical reasoning.
			In $\SCUi$, $(\IYB_\omega^n)$ and $(\4^*)$ give
			$\UO_i^\omega\phi\to \K_i\UO_i^n\phi$ for every $n\in\mathbb N^{+}$.
			By modal reasoning, these yield the premises
			$\UO_i^\omega\phi\to \K_i(\top\to\UO_i^n\phi)$ of the $m=1$ instance of
			$(\omega \mathrm{I})$, which gives
			$\UO_i^\omega\phi\to \K_i(\top\to\UO_i^\omega\phi)$, and hence $(\4^*_\omega)$.
			Then $(\5^*_\omega)$ follows by the same argument as $(\5^*)$.
			Given $(\Irref)$ and $(\Trans)$, the derivation of $(\Asym)$ is routine. For $(\Irref)$ in $\SCUzero$:
			$$
			\begin{array}{lll}
				(1) & (i\succ i)\phi \rightarrow \UO_i^1\phi & (\CMP1)\\
				(2_m) & (i\succ i)\phi \wedge \UO_i^m\phi \rightarrow \UO_i^{m+1}\phi
				& (\CMPn)\qquad (1\leqslant m<\kappa)\\
				(3) & (i\succ i)\phi \rightarrow \UO_i^\kappa\phi
				& (1),(2_m),\text{ induction on }m\ (1\leqslant m<\kappa)\\
				(4) & (i\succ i)\phi \rightarrow \neg \UO_i^\kappa\phi & (\CMP\kappa)\\
				(5) & \neg(i\succ i)\phi & (3),(4),\text{ classical reasoning}
			\end{array}
			$$
			
				For $(\Trans)$ in $\SCUzero$, write
				$A:=(i\succ j)\phi\wedge (j\succ k)\phi$ and $B:=(i\succ k)\phi$. Then:
				$$
				\begin{array}{lll}
					(1) & A\rightarrow \neg \UO_k^\kappa\phi
					& (\CMP\kappa),\text{ propositional reasoning}\\
					(2) & A\rightarrow \UO_i^1\phi\wedge \K_k\phi
					& (\CMP1),\text{ propositional reasoning}\\
					(3) & A\wedge \neg \UO_k^1\phi\rightarrow B
					& (2),(\UYC0)\\[1mm]
					
					(4_m) & (j\succ k)\phi\wedge \UO_k^m\phi \rightarrow \UO_j^{m+1}\phi
					& (\CMPn)\qquad (1\leqslant m<\kappa)\\
					(5_m) & \UO_j^{m+1}\phi\rightarrow \UO_j^m\phi
					& (\IYB)\qquad (1\leqslant m<\kappa)\\
					(6_m) & (i\succ j)\phi\wedge \UO_j^m\phi \rightarrow \UO_i^{m+1}\phi
					& (\CMPn)\qquad (1\leqslant m<\kappa)\\
					(7_m) & A\wedge \UO_k^m\phi \rightarrow \UO_i^{m+1}\phi
					& (4_m),(5_m),(6_m)\\
						(8_m) & A\wedge \UO_k^m\phi\wedge \neg \UO_k^{m+1}\phi
						\rightarrow B
						& (7_m),(\UYC)\qquad (1\leqslant m<\kappa)\\
						(9_m) & A\wedge\neg B\wedge \UO_k^m\phi
						\rightarrow \UO_k^{m+1}\phi
						& (8_m),\text{ classical reasoning}\\[1mm]
						
						(10) & A\wedge\neg B\rightarrow \UO_k^\kappa\phi
						& (3),(9_m),\text{ induction on }m\\
					(11) & A\rightarrow B
					& (1),(10),\text{ classical reasoning}
				\end{array}
				$$

		For $(\Irref)$ in $\SCUi$:
		$$
		\begin{array}{lll}
			(1) & (i\succ i)\phi\rightarrow \UO_i^1\phi & (\CMP1)\\
			(2_m) & (i\succ i)\phi \wedge \UO_i^m\phi \rightarrow \UO_i^{m+1}\phi
			& (\CMPn)\qquad (m\in\mathbb N^{+})\\
			(3_n) & (i\succ i)\phi\rightarrow \UO_i^n\phi
			& (1),(2_m),\text{ induction on }n\\
			(4) & (i\succ i)\phi\rightarrow \UO_i^\omega\phi
			& (\omega \mathrm{I})\text{ applied to }(3_n)_{n\in\mathbb N^{+}}\\
			
			(5) & (i\succ i)\phi\rightarrow \neg \UO_i^\omega\phi & (\CMP\omega)\\
			(6) & \neg(i\succ i)\phi & (4),(5),\text{ classical reasoning}
		\end{array}
		$$

				For $(\Trans)$ in $\SCUi$, use the same abbreviations
				$A:=(i\succ j)\phi\wedge (j\succ k)\phi$ and $B:=(i\succ k)\phi$. Then:
				$$
				\begin{array}{lll}
					(1) & A\rightarrow \neg \UO_k^\omega\phi & (\CMP\omega),\text{ propositional reasoning}\\
					(2) & A\rightarrow \UO_i^1\phi\wedge \K_k\phi & (\CMP1),\text{ propositional reasoning}\\
					(3) & A\wedge \neg \UO_k^1\phi\rightarrow B & (2),(\UYC0)\\[1mm]
					
					(4_m) & (j\succ k)\phi\wedge \UO_k^m\phi \rightarrow \UO_j^{m+1}\phi
					& (\CMPn)\qquad (m\in\mathbb N^{+})\\
					(5_m) & \UO_j^{m+1}\phi\rightarrow \UO_j^m\phi
					& (\IYB)\qquad (m\in\mathbb N^{+})\\
					(6_m) & (i\succ j)\phi\wedge \UO_j^m\phi \rightarrow \UO_i^{m+1}\phi
					& (\CMPn)\qquad (m\in\mathbb N^{+})\\
					(7_m) & A\wedge \UO_k^m\phi \rightarrow \UO_i^{m+1}\phi
					& (4_m),(5_m),(6_m)\\
					(8_m) & A\wedge \UO_k^m\phi\wedge \neg \UO_k^{m+1}\phi\rightarrow B
					& (7_m),(\UYC)\qquad (m\in\mathbb N^{+})\\
					(9_m) & A\wedge\neg B\wedge \UO_k^m\phi
					\rightarrow \UO_k^{m+1}\phi
					& (8_m),\text{ classical reasoning}\\[1mm]
					
					(10_n) & A\wedge\neg B\rightarrow \UO_k^n\phi
					& (3),(9_m),\text{ induction on }n\\
					(11) & A\wedge\neg B\rightarrow \UO_k^\omega\phi
					& (\omega \mathrm{I})\text{ applied to }(10_n)_{n\in\mathbb N^{+}}\\
					(12) & A\rightarrow B & (1),(11),\text{ classical reasoning}
				\end{array}
				$$
		\end{proof}

		\begin{theorem}[Soundness for $\SCUzero$]\label{thm:soundness-fin}
			Fix $\kappa\geqslant2$.
			$\SCUzero$ is sound over $\ELBU^\kappa$-models.
		\end{theorem}
		\begin{proof}[Proof sketch]
			Below we record only two representative nontrivial cases.
			
			
			For $(\KYU)$, assume $\M,w\vDash \UO_i^1\K_i\phi$, witnessed by $t\in E_1$.
			Then $\gr(t)\geqslant1$, so $\gr(!t)=2$ and hence $!t\in E_2$.
			For each $v$ with $wR_iv$, we have $v\in\Ex_i(t,\K_i\phi)$ and $\M,v\vDash\K_i\phi$;
			by epistemic introspection of $\Ex_i$, also $v\in\Ex_i(!t,\K_i\phi)$.
			Thus $!t$ witnesses $\M,w\vDash \UO_i^2\K_i\phi$.
			
			For $(\CMP\kappa)$, if $\M,w\vDash(i\succ j)\phi\wedge\UO_j^\kappa\phi$, then
			$\deg_j(w,\phi)=\kappa$ while over $\ELBU^\kappa$ always $\deg_i(w,\phi)\leqslant\kappa$,
			contradicting $\deg_i(w,\phi)>\deg_j(w,\phi)$; hence $\M,w\vDash\neg\UO_j^\kappa\phi$.
		\end{proof}
		
		\begin{theorem}[Soundness for $\SCUi$]\label{thm:soundness-inf}
			$\SCUi$ is sound over $\ELBU$ models.
		\end{theorem}
	\begin{proof}[Proof sketch]
			By induction on $\SCUi$ proof trees from Definition~\ref{def:scui-proof}. The shared finitary axioms/rules are sound as in Theorem~\ref{thm:soundness-fin}.
			For $(\IYB_\omega^n)$, assume $\M,w\vDash\UO_i^\omega\phi$.
			By the semantic clause of $\UO^\omega$, this means $\M,w\vDash\UO_i^n\phi$ for every $n\geqslant1$.
			$(\CMP\omega)$ is valid because $(i\succ j)\phi$ requires $\deg_i(w,\phi)>\deg_j(w,\phi)$, impossible if $\UO_j^\omega\phi$ holds (then $\deg_j(w,\phi)=\omega$).
			For $(\omega \mathrm{I})$, assume all its premises are
			true at $w$. If $\M,w\not\vDash\chi$, the conclusion is immediate. Suppose
			$\M,w\vDash\chi$. If $m=0$, the premises give
			$\M,w\vDash\UO_i^n\phi$ for every $n\geqslant1$, hence
			$\M,w\vDash\UO_i^\omega\phi$ by the semantics.
			If $m>0$, take arbitrary worlds $w_1,\dots,w_m$ with
			$wR_{j_1}w_1,\dots,w_{m-1}R_{j_m}w_m$. If there is a first $r$ such that
			$\theta_r$ fails at $w_r$, the corresponding implication is true. If all $\theta_r$ hold, then the
			premises give $\M,w_m\vDash\UO_i^n\phi$ for every $n\geqslant1$, so
			$\M,w_m\vDash\UO_i^\omega\phi$. Thus the conclusion is true at $w$.
		\end{proof}
		
		\section{Completeness and Bounded Decidability}\label{sec:completeness}
		This section proves strong completeness for the bounded finitary and full infinitary
		calculi, and also establishes bounded decidability. We therefore use two
		canonical-model constructions: a bounded one, following the witness-based strategy of
		\cite{xu2016logic,yu2024understanding} but adapted to graded witnesses, agent-indexed
		explanation stores, and comparative formulas, and a full-language construction based on maximal consistent sets closed under the infinitary rule.


	We first construct the bounded canonical model for $\SCUzero$.	Let $\Omega$ be the set of all maximal $\SCUzero$-consistent sets of
		$\ELBU^\kappa$-formulas.
		
		\begin{definition}\label{def:cm-xu}
			
		The canonical model  for $\SCUzero$ is
		$\M^c=(W^c,\{R_i^c\mid i\in I\},V^c,E^c,\gr^c,\{\Ex_i^c\mid i\in I\})$, where:
		
			\begin{itemize}
				\item $E^c$ is generated by the BNF:
				$	t ::= c \mid \varphi^n \mid (t\cdot t)\mid (t+t)\mid !t$,
				where $\varphi\in\ELBU^\kappa$, $n\in\{1,\dots,\kappa\}$.
			
				\item The canonical grade map $\gr^c:E^c\to\mathbb N$ is defined inductively by:
				\[
				\begin{array}{rcl}
					\gr^c(c)&=&1,\\
					\gr^c(\varphi^n)&=&n,\\
					\gr^c(t\cdot s)&=&\min\{\gr^c(t),\gr^c(s)\},\\[1mm]
					\gr^c(!t)&=&
					\begin{cases}
						0, & \text{if }\gr^c(t)=0,\\
						2, & \text{if }\gr^c(t)>0,
					\end{cases}\\[3mm]
					\gr^c(t+s)&=&
					\begin{cases}
						0, & \text{if }\min\{\gr^c(t),\gr^c(s)\}=0,\\
						\min\{\gr^c(t),\gr^c(s)\}-1, & \text{otherwise.}
					\end{cases}
				\end{array}
				\]
				For each $n\in\mathbb N$, let $E_n^c:=\{t\in E^c\mid \gr^c(t)\geqslant n\}$.
				
			\item $W^c$ is the set of all triples $\langle\Gamma,\mathcal F,\vec f\rangle$
			such that $\Gamma\in\Omega$ and:
				\begin{itemize}
					\item $\vec f=\{f_i^n\mid i\in I,\ n\in\{1,\dots,\kappa\}\}$ with
					$f_i^n:\{\varphi\mid \UO_i^n\varphi\in\Gamma\}\to \{t\in E^c\mid \gr^c(t)=n\}$;
					\item $\mathcal F=\{F_i\mid i\in I\}$, where each
					$F_i\subseteq E^c\times \ELBU^\kappa$ is the \emph{agent-$i$ explanation store}
					satisfying the following base and admissibility conditions relative to
					$\Gamma$ and $\vec f$:
					\[
					\begin{array}{@{}lp{0.78\linewidth}@{}}
						\textit{(Base$_i$)} &
						$\{\langle c,\phi\rangle\mid \phi\in\Lambda\}\subseteq F_i$ and
						$\{\langle f_i^n(\varphi),\varphi\rangle\mid \UO_i^n\varphi\in\Gamma\}\subseteq F_i$;\\[1mm]
						\textit{(App$_i$)} &
						if $\langle t,\varphi\to\psi\rangle,\langle s,\varphi\rangle\in F_i$ then $\langle t\cdot s,\psi\rangle\in F_i$;\\[1mm]
						\textit{(Sum$_i$)} &
						if $\langle t,\varphi\rangle\in F_i$ or $\langle s,\varphi\rangle\in F_i$ then $\langle t+s,\varphi\rangle\in F_i$;\\[1mm]
						\textit{(Int$_i$)} &
						if $\langle t,\K_i\varphi\rangle\in F_i$, then $\langle !t,\K_i\varphi\rangle\in F_i$.
					\end{array}
					\]
				\end{itemize}
				
				\item 
				$\langle\Gamma,\mathcal F,\vec f\rangle R_i^c \langle\Delta,\mathcal G,\vec g\rangle
				$ iff $
				\{\varphi\mid \K_i\varphi\in\Gamma\}\subseteq\Delta\ $ and $
				\forall n\in\{1,\dots,\kappa\}\ (f_i^n=g_i^n)$.
				
				\item For each $i\in I$, $\Ex_i^c: E^c\times \ELBU^\kappa\to 2^{W^c}$ is defined by
				$	\Ex_i^c(t,\varphi):=\{\langle\Gamma,\mathcal F,\vec f\rangle\in W^c \mid \langle t,\varphi\rangle\in F_i\}$.

				\item $V^c(p):=\{\langle\Gamma,\mathcal F,\vec f\rangle\in W^c\mid p\in \Gamma\}$.
			\end{itemize}
		\end{definition}
		
		
		In this construction, each canonical world $\langle \Gamma,F,\vec{f}\rangle\in W^c$
		packages exact-grade witnesses for the $\UO$-formulas in $\Gamma$.
		For each $\UO_i^n\varphi\in\Gamma$, the map $f_i^n$ selects a term
		$f_i^n(\varphi)\in E^c$ with $\gr^c(f_i^n(\varphi))=n$, and the pair
		$\langle f_i^n(\varphi),\varphi\rangle$ is placed into the base of $F_i$.
		Each $F_i$ is then closed under \textup{(App$_i$)}, \textup{(Sum$_i$)}, and \textup{(Int$_i$)}.
		The specific clauses in the canonical grading $\gr^c$ are chosen so that the canonical model
		still satisfies the original semantic constraints and, at the same time, supports the later
		technical arguments.

		
		We next define the closure operation used to build explanation stores from
		their base pairs. It adds exactly the pairs forced by the admissibility
		conditions for application, sum, and epistemic introspection.
		\begin{definition}\label{def:iter-closure}
			For $X\subseteq E^c\times\ELBU^\kappa$ and $i\in I$, define the one-step closure operator 	$\mathsf{cl}^{\,i}(X)\subseteq E^c\times\ELBU^\kappa$ by
			$	\mathsf{cl}^{\,i}(X):=X\cup \mathsf{cl}_{\mathrm{app}}(X)\cup \mathsf{cl}_{\mathrm{sum}}(X)\cup \mathsf{cl}^{\,i}_{\mathrm{int}}(X)$,	where
			\begin{align*}
				\mathsf{cl}_{\mathrm{app}}(X)
				&:=\{\langle t\cdot s,\psi\rangle \mid \exists \chi\ (\langle t,\chi\to\psi\rangle\in X\ \wedge\ \langle s,\chi\rangle\in X)\},\\
				\mathsf{cl}_{\mathrm{sum}}(X)
				&:=\{\langle t+s,\varphi\rangle \mid \langle t,\varphi\rangle\in X\ \text{or}\ \langle s,\varphi\rangle\in X\},\\
				\mathsf{cl}^{\,i}_{\mathrm{int}}(X)
				&:=\{\langle !t,\K_i\varphi\rangle \mid \langle t,\K_i\varphi\rangle\in X\}.
			\end{align*}
			
			Given $S\subseteq E^c\times\ELBU^\kappa$, let $S_i^0:=S$ and
			$S_i^{k+1}:=\mathsf{cl}^{\,i}(S_i^k)$ for $k\in\mathbb N$.
			Define $S_i^\infty:=\bigcup_{k\in\mathbb N}S_i^k$.

		\end{definition}
		
		\begin{lemma}\label{lem:canonical-realization}
			For every maximal $\SCUzero$-consistent set $\Gamma\in\Omega$,
			there exist $\mathcal F=\{F_i\mid i\in I\}$ and a witness family $\vec f$
			such that $\langle\Gamma,\mathcal F,\vec f\rangle\in W^c$.
			In particular, $W^c\neq\emptyset$.
		\end{lemma}
		\begin{proof}
			For each $i\in I$ and $n\in \{1,\dots,\kappa\}$, define
			$f_i^n(\varphi):=\varphi^n
			\text{ for }\varphi\text{ with }\UO_i^n\varphi\in\Gamma$.
			This is well-typed since $\gr^c(\varphi^n)=n$.
			For each fixed $i\in I$, let
			\[
			S_{\Gamma,\vec f,i}
			:=\{\langle c,\phi\rangle\mid \phi\in\Lambda\}\ \cup\
			\{\langle f_i^n(\varphi),\varphi\rangle\mid n\in\{1,\dots,\kappa\},\ \UO_i^n\varphi\in\Gamma\}.
			\]
			
			Using the notation of Definition~\ref{def:iter-closure}, set $S_{\Gamma,\vec f,i}^{0}:=S_{\Gamma,\vec f,i}$,
			$S_{\Gamma,\vec f,i}^{k+1}:=\mathsf{cl}^{\,i}(S_{\Gamma,\vec f,i}^{k})$,
			$F_i:=S_{\Gamma,\vec f,i}^{\infty}$.
			Let $\mathcal F:=\{F_i\mid i\in I\}$. Then each $F_i$ is base-extending and closed under \textup{(App$_i$)}, \textup{(Sum$_i$)}, and \textup{(Int$_i$)} by construction.
			Hence $\langle\Gamma,\mathcal F,\vec f\rangle\in W^c$.
			
			For nonemptiness, take any theorem $\top$ of $\SCUzero$.
			Then $\{\top\}$ is $\SCUzero$-consistent, so by Lindenbaum there exists
			$\Gamma\in\Omega$. Applying the first part to this $\Gamma$ gives
			$\langle\Gamma,\mathcal F,\vec f\rangle\in W^c$. Hence $W^c\neq\emptyset$.
		\end{proof}
		\begin{lemma}\label{lem:Rc-equivalence-xu}
			For each $i\in I$, $R_i^c$ is an equivalence relation on $W^c$.
		\end{lemma}

	We omit the proof.	Regarding $\Ex_i^c$ in canonical models, it is obvious to check the following:
		\begin{lemma}\label{lem:exi-adm-c}
			For each $i\in I$, the function $\Ex_i^c$ satisfies all conditions in
			Definition~\ref{cu} for $L=\{1,\dots,\kappa\}$.
		\end{lemma}
	
		We conclude that the canonical model is well-defined, based on the above lemmas.
		\begin{proposition}\label{prop:cm-well-defined}
			$\M^c$ is a $\ELBU^\kappa$-model (equivalently, a $\ELBU(L)$-model with
			$L=\{1,\dots,\kappa\}$).
		\end{proposition}
		\begin{proof}
			$E^c$ is nonempty and closed under $\cdot,+, !$ by its BNF generation.
			The clauses defining $\gr^c$ satisfy all grade constraints in Definition~\ref{cu}.
			By Lemma~\ref{lem:canonical-realization}, $W^c\neq\emptyset$.
			By Lemma~\ref{lem:Rc-equivalence-xu}, each $R_i^c$ is an equivalence relation on $W^c$.
			By Lemma~\ref{lem:exi-adm-c}, each $\Ex_i^c$ satisfies the required admissibility clauses.
		\end{proof}
		
		We now establish the existence ingredients for the Truth Lemma.
		
		\begin{lemma}[$\K_i$-Existence Lemma]\label{k}
			For any canonical world $w=\langle \Gamma, \mathcal F, \vec f\rangle\in W^c$, if $\neg \K_i\phi \in \Gamma$, then there exists
			a canonical world $w'=\langle\Delta, \mathcal G, \vec g\rangle\in W^c$ such that
			$wR_i^c w'$ and $\neg\phi\in \Delta$.
		\end{lemma}
		\begin{proof}
			Fix $w=\langle \Gamma, \mathcal F, \vec f\rangle\in W^c$ with $\neg \K_i\phi\in\Gamma$. Let
			$\Delta^- \ :=\ \{\psi\mid \K_i\psi\in \Gamma\}\ \cup\ \{\neg\phi\}$.
			It is routine to show that $\Delta^-$ is $\SCUzero$-consistent, by $(\Nn)$ and $(\Kg)$.
			Extend $\Delta^-$ to a maximal $\SCUzero$-consistent set $\Delta\in\Omega$.
			Then $\{\psi\mid \K_i\psi\in\Gamma\}\subseteq \Delta$ and $\neg\phi\in\Delta$.
			
			For each $n\in\{1,\dots,\kappa\}$ and formula $\chi$, we have
			$\UO_i^n\chi\in\Gamma\Leftrightarrow \UO_i^n\chi\in\Delta$.
			Indeed, if $\UO_i^n\chi\in\Gamma$, then $(\4^*)$ gives $\K_i\UO_i^n\chi\in\Gamma$,
			so $\UO_i^n\chi\in\Delta$.
			Conversely, if $\UO_i^n\chi\notin\Gamma$, then $\neg\UO_i^n\chi\in\Gamma$ by maximality;
			by theorem $(\5^*)$, $\K_i\neg\UO_i^n\chi\in\Gamma$, hence
			$\neg\UO_i^n\chi\in\Delta$, so $\UO_i^n\chi\notin\Delta$.
			Thus setting $g_i^n:=f_i^n$ is type-correct for the $\Delta$-domain requirement in $W^c$.
			
			Define $\vec g=\{g_j^n\mid j\in I,\ n\in\{1,\dots,\kappa\}\}$ by:
			for $j=i$, let $g_i^n:=f_i^n$; for $j\neq i$, let $g_j^n(\psi):=\psi^n$ whenever
			$\UO_j^n\psi\in\Delta$.
			This is well-defined because $\psi^n\in E^c$ and $\gr^c(\psi^n)=n$.
			
			To construct $\mathcal G=\{G_j\mid j\in I\}$, for each $j\in I$ define
			\[
			S_{\Delta,\vec g,j}:=\{\langle c,\phi\rangle\mid \phi\in\Lambda\}\ \cup\
			\{\langle g_j^n(\psi),\psi\rangle\mid n\in\{1,\dots,\kappa\},\ \UO_j^n\psi\in\Delta\},
			\]
			and then let $S_{\Delta,\vec g,j}^{0}:=S_{\Delta,\vec g,j}$,
			$S_{\Delta,\vec g,j}^{m+1}:=\mathsf{cl}^{\,j}(S_{\Delta,\vec g,j}^{m})$,
			$S_{\Delta,\vec g,j}^{\infty}:=\bigcup_{m\in\mathbb N}S_{\Delta,\vec g,j}^{m}$, and
			$G_j:=S_{\Delta,\vec g,j}^{\infty}$.
			Each $G_j$ is base-extending and closed under \textup{(App$_j$)}, \textup{(Sum$_j$)}, and \textup{(Int$_j$)}.
			Set $\mathcal G:=\{G_j\mid j\in I\}$ and let $w':=\langle\Delta,\mathcal G,\vec g\rangle$.
			Then $w'\in W^c$.
			
			By the definition of $R_i^c$, since $\{\psi\mid \K_i\psi\in\Gamma\}\subseteq \Delta$ and
			$f_i^n=g_i^n$ for all $n\in\{1,\dots,\kappa\}$, we have $wR_i^c w'$.
			Finally, $\neg\phi\in\Delta$ by construction.
		\end{proof}
		
		\begin{lemma}[$\UO_i^n$-Existence Lemma]\label{lem:uo-refute-transfer}
			Let $w=\langle \Gamma,\mathcal F,\vec f\rangle\in W^c$, fix $i\in I$ and $n\in\{1,\dots,\kappa\}$, and assume
			$\UO_i^n\phi\notin\Gamma$ and $\langle t,\phi\rangle\in F_i$ for some $t\in E_n^c$.
			There exists $w'=\langle\Delta,\mathcal G,\vec g\rangle\in W^c$ such that
			$wR_i^c w'$ and $\langle t,\phi\rangle\notin G_i$.
		\end{lemma}
		\begin{proof}
			Fix $w=\langle \Gamma,\mathcal F,\vec f\rangle$.
			Define, for this proof:
			\[
			C^0:=\{\langle c,\phi\rangle\mid \phi\in\Lambda\}\ \cup\
			\{\langle f_i^m(\psi),\psi\rangle\mid m\in\{1,\dots,\kappa\},\ \UO_i^m\psi\in\Gamma\},
			\]
		and $C^{k+1}:=\mathsf{cl}^{\,i}(C^k)$,	$C^\infty:=\bigcup_{k\in\mathbb N}C^k$. Thus $C^\infty$ is the least agent-$i$ explanation store generated from
		$\Gamma$ and the witness maps $\vec f$ by the admissibility conditions.

			\textbf{Claim 1.}
			If $\langle t,\phi\rangle\in C^\infty$ and $\gr^c(t)\geqslant n\geqslant 1$, then
			$\UO_i^n\phi\in\Gamma$.
			
			\textbf{Claim 2.}
			Given the fixed $w=\langle \Gamma,\mathcal F,\vec f\rangle\in W^c$, $i\in I$ and $t\in E^c$, if
			$\langle t,\phi\rangle\notin C^\infty$, then there exists
			$w'=\langle\Delta,\mathcal G,\vec g\rangle\in W^c$ such that
			$wR_i^c w'$
			and $\langle t,\phi\rangle\notin G_i$.

			\smallskip
			\noindent\emph{Proof of Claim 1.}
			We prove by induction on $k\in\mathbb N$ that each pair in $C^k$ has, in
			$\Gamma$, all positive understanding formulas up to the grade of its term:
			\[
			P(k):\ \forall\langle u,\psi\rangle\in C^k\ \forall \ell\
			(1\leqslant \ell\leqslant \gr^c(u)\Rightarrow \UO_i^\ell\psi\in\Gamma).
			\]

			For $k=0$, there are two kinds of pairs in $C^0$.
			If $\langle u,\psi\rangle=\langle c,\lambda\rangle$ with $\lambda\in\Lambda$, then
			$\gr^c(c)=1$ and $\UO_i^1\lambda$ is derivable by $(\NE)$, hence belongs to $\Gamma$.
			If $\langle u,\psi\rangle=\langle f_i^m(\psi),\psi\rangle$, then $\UO_i^m\psi\in\Gamma$ by
			definition of the domain of $f_i^m$; for any $1\leqslant \ell\leqslant m$, $(\IYB)$ yields
			$\UO_i^\ell\psi\in\Gamma$.

			For the induction step, take $\langle u,\psi\rangle\in C^{k+1}=\mathsf{cl}^{\,i}(C^k)$ and
			$1\leqslant \ell\leqslant \gr^c(u)$. If $\langle u,\psi\rangle\in C^k$, apply IH. Otherwise:
			\begin{description}
				\item[\textup{(App)}]
				$u=r\cdot s$ with $\langle r,\chi\to\psi\rangle,\langle s,\chi\rangle\in C^k$.
				Let $a:=\gr^c(r)$, $b:=\gr^c(s)$, and $h:=\min\{a,b\}$; then $\gr^c(u)=h$.
				By IH and $(\IYB)$, $\UO_i^h(\chi\to\psi),\UO_i^h\chi\in\Gamma$.
				By $(\Kgg)$, $\UO_i^{h}\psi\in\Gamma$, and then $(\IYB)$ yields
				$\UO_i^\ell\psi\in\Gamma$.

				\item[\textup{(Sum)}]
				$u=r+s$ and either $\langle r,\psi\rangle\in C^k$ or $\langle s,\psi\rangle\in C^k$.
				Assume $\langle r,\psi\rangle\in C^k$.
				Set $a:=\gr^c(r)$ and $h:=\gr^c(u)$.
				Since $1\leqslant \ell\leqslant h$, we have $h\geqslant1$.
				By the grade definition for $+$, this yields $h<a$, hence $\ell<a$.
				By IH, $\UO_i^a\psi\in\Gamma$, and then $(\IYB)$ gives $\UO_i^\ell\psi\in\Gamma$.

				\item[\textup{(Int)}]
				$u=!r$ with $\langle r,\K_i\psi\rangle\in C^k$.
				If $\gr^c(r)=0$, then $\gr^c(u)=0$, impossible since $1\leqslant \ell\leqslant \gr^c(u)$.
				If $\gr^c(r)\geqslant1$, then $\gr^c(u)=2$.
				By IH, $\UO_i^1\K_i\psi\in\Gamma$.
				By $(\KYU)$, $\UO_i^2\K_i\psi\in\Gamma$, and then $(\IYB)$ yields
				$\UO_i^\ell\K_i\psi\in\Gamma$.
			\end{description}
			So $P(k+1)$ holds. Hence $P(k)$ holds for all $k$. If $\langle t,\phi\rangle\in C^\infty$, then
			$\langle t,\phi\rangle\in C^k$ for some $k$, and applying $P(k)$ with $\ell:=n$
			(since $1\leqslant n\leqslant \gr^c(t)$) yields $\UO_i^n\phi\in\Gamma$.

			\smallskip
			\noindent\emph{Proof of Claim 2.}
			Set $\Delta:=\Gamma$, and for each
			$m\in\{1,\dots,\kappa\}$ let $g_i^m:=f_i^m$. For $j\neq i$, define
			$g_j^m(\psi):=\psi^m$ for $\UO_j^m\psi\in\Delta$. Then $\vec g$ satisfies the typing condition
			in the definition of $W^c$, and $g_i^m=f_i^m$ for all $m$.

			Let $G_i:=C^\infty$. By hypothesis, $\langle t,\phi\rangle\notin G_i$.
			For each $j\neq i$, let
			\[
			S_{\Delta,\vec g,j}:=\{\langle c,\phi\rangle\mid \phi\in\Lambda\}\ \cup\
			\{\langle g_j^m(\psi),\psi\rangle\mid m\in\{1,\dots,\kappa\},\ \UO_j^m\psi\in\Delta\},
			\]
			and define
			$S_{\Delta,\vec g,j}^{0}:=S_{\Delta,\vec g,j}$,
			$S_{\Delta,\vec g,j}^{k+1}:=\mathsf{cl}^{\,j}(S_{\Delta,\vec g,j}^{k})$,
			$S_{\Delta,\vec g,j}^{\infty}:=\bigcup_{k\in\mathbb N}S_{\Delta,\vec g,j}^{k}$, and
			$G_j:=S_{\Delta,\vec g,j}^{\infty}$.
			Let $\mathcal G:=\{G_j\mid j\in I\}$ and define $w':=\langle\Delta,\mathcal G,\vec g\rangle$.
			For $j=i$, $G_i=C^\infty$ is base-extending and closed under
			\textup{(App$_i$)}, \textup{(Sum$_i$)}, \textup{(Int$_i$)} by construction.
			For $j\neq i$, each $G_j$ has the same closure properties by the iterative construction above.
			Hence $w'\in W^c$.
			Also, $\{\psi\mid \K_i\psi\in\Gamma\}\subseteq\Delta$ and $f_i^m=g_i^m$ for all $m$, so
			$wR_i^c w'$.
			Finally, $\langle t,\phi\rangle\notin G_i$.
			
			Now we conclude the lemma.
			If $\langle t,\phi\rangle\in C^\infty$, then Claim 1 gives
			$\UO_i^n\phi\in\Gamma$, contradiction.
			Hence $\langle t,\phi\rangle\notin C^\infty$.
			By Claim 2 there exists an $i$-successor
			$w'=\langle\Delta,\mathcal G,\vec g\rangle\in W^c$ with $\langle t,\phi\rangle\notin G_i$.
		\end{proof}
		
	With these existence lemmas in place, we prove the bounded Truth Lemma.

		\begin{lemma}[Truth Lemma for $\ELBU^\kappa$]\label{lem:truth-ku}
			For $\varphi\in\ELBU^\kappa$ and
			$w=\langle\Gamma,\mathcal F,\vec f\rangle\in W^c$,
		$\M^c,w\vDash\varphi$ iff $\varphi\in\Gamma$.
		\end{lemma}
		\begin{proof}
			We use induction on formula structure.
			The Boolean cases are standard. For $\varphi=\K_i\psi$:
			\begin{itemize}
				\item
				Assume $\K_i\psi\in\Gamma$, and let
				$w'=\langle\Delta,\mathcal G,\vec g\rangle$ satisfy
				$wR_i^c w'$.
				By definition of $R_i^c$, $\{\chi\mid \K_i\chi\in\Gamma\}\subseteq\Delta$, hence $\psi\in\Delta$.
				By IH, $\M^c,w'\vDash\psi$.
				So $\M^c,w\vDash \K_i\psi$.
				\item
				Assume $\M^c,w\vDash \K_i\psi$.
				If $\K_i\psi\notin\Gamma$, then $\neg \K_i\psi\in\Gamma$ by maximality.
				By Lemma~\ref{k}, there is $w'=\langle\Delta,\mathcal G,\vec g\rangle$ with
				$wR_i^c w'$
				and $\neg\psi\in\Delta$.
				By IH, $\M^c,w'\not\vDash\psi$, contradiction.
				Hence $\K_i\psi\in\Gamma$.
			\end{itemize}
			
			For $\varphi=\UO_i^n\psi$:
			\begin{itemize}
				\item
				Assume $\UO_i^n\psi\in\Gamma$.
				Let $t:=f_i^n(\psi)$; then $\gr^c(t)=n$, so $t\in E_n^c$.
				We show that $t$ witnesses $\UO_i^n\psi$ at $w$. Let $w'=\langle\Delta,\mathcal G,\vec g\rangle$ satisfy $wR_i^c w'$.
				By the definition of $R_i^c$, $g_i^n=f_i^n$. Since
				$\psi\in\mathrm{dom}(f_i^n)$, we have $\psi\in\mathrm{dom}(g_i^n)$, hence
				$\UO_i^n\psi\in\Delta$. The base condition for $G_i$ gives
				$\langle g_i^n(\psi),\psi\rangle\in G_i$, and therefore
				$\langle t,\psi\rangle\in G_i$. Thus $w'\in\Ex_i^c(t,\psi)$. From $(\IMP)$ and $\UO_i^n\psi\in\Gamma$, we get $\K_i\psi\in\Gamma$.
				By the $\K_i$-case already proved, $\M^c,w\vDash\K_i\psi$.
				Therefore $\M^c,w\vDash\UO_i^n\psi$.
				
				\item
				Assume $\M^c,w\vDash\UO_i^n\psi$. Then there exists $t\in E_n^c$ such that
				$\forall v\in W^c\;(	wR_i^c v\Rightarrow v\in\Ex_i^c(t,\psi))$. Since $R_i^c$ is reflexive, $w\in\Ex_i^c(t,\psi)$,
				i.e.\ $\langle t,\psi\rangle\in F_i$. If $\UO_i^n\psi\notin\Gamma$, then Lemma~\ref{lem:uo-refute-transfer}
				gives $w'=\langle\Delta,\mathcal G,\vec g\rangle$ with $wR_i^c w'$ and
				$\langle t,\psi\rangle\notin G_i$. Hence
				$w'\notin\Ex_i^c(t,\psi)$, contradiction.
				Therefore $\UO_i^n\psi\in\Gamma$.
				
			\end{itemize}
			
			For $\varphi=(i\succ j)\psi$. First handle the special case $i=j$.
			By theorem $(\Irref)$, no $\Gamma$ contains $(i\succ i)\psi$.
			Semantically, $\M^c,w\vDash(i\succ i)\psi$ is impossible since it would require
			$\deg_i(w,\psi)>\deg_i(w,\psi)$.
			Hence, $\M^c,w\vDash(i\succ i)\psi
			\Leftrightarrow
			(i\succ i)\psi\in\Gamma$. Now assume $i\neq j$.
			\begin{itemize}
				\item Assume $(i\succ j)\psi\in\Gamma$.
				From $(\CMP1)$, $\K_j\psi\in\Gamma$ and $\UO_i^1\psi\in\Gamma$.
				From $(\CMP\kappa)$, $\neg\UO_j^\kappa\psi\in\Gamma$.
				By IH, $\M^c,w\vDash\K_j\psi$.
				Let $S_j:=\{n\leqslant\kappa\mid \UO_j^n\psi\in\Gamma\}$.
				By $(\IYB)$, $S_j$ is downward closed.
				Since $\neg\UO_j^\kappa\psi\in\Gamma$, we have $S_j\neq\{1,\dots,\kappa\}$.
				If $S_j=\emptyset$, then
				$\deg_j(w,\psi)=0$, while $\UO_i^1\psi\in\Gamma$ gives
				$\deg_i(w,\psi)\geqslant1$ by IH.
				If $S_j\neq\emptyset$, let $m+1\leqslant\kappa$ be the least  not in $S_j$.
				Then $m\leqslant\kappa-1$ and $S_j=\{1,\dots,m\}$.
				Hence $\UO_j^m\psi\in\Gamma$ and $\UO_j^{m+1}\psi\notin\Gamma$; then
				$\neg\UO_j^{m+1}\psi\in\Gamma$, and by $(\CMPn)$,
				$\UO_i^{m+1}\psi\in\Gamma$.
				By IH,
				$\deg_j(w,\psi)=m$
				and
				$\deg_i(w,\psi)\geqslant m+1$.
				In either case,
				$\deg_i(w,\psi)>
				\deg_j(w,\psi)$.
				So $\M^c,w\vDash(i\succ j)\psi$.
				\item Assume
				$\M^c,w\vDash(i\succ j)\psi$.
				Then $\M^c,w\vDash\K_j\psi$, so $\K_j\psi\in\Gamma$ by IH.
				If $\deg_j(w,\psi)=0$, then
				$\deg_i(w,\psi)\geqslant1$, hence $\M^c,w\vDash\UO_i^1\psi$, and
				$\M^c,w\vDash\neg\UO_j^1\psi$; by IH, both belong to $\Gamma$, hence
				$(i\succ j)\psi\in\Gamma$ by $(\UYC0)$.
				If $\deg_j(w,\psi)=m\geqslant1$, then necessarily $m\leqslant\kappa-1$.
				Then
				$\M^c,w\vDash\UO_j^m\psi\wedge\neg\UO_j^{m+1}\psi$
				and
				$\M^c,w\vDash\UO_i^{m+1}\psi$.
				By IH these formulas are in $\Gamma$, and $(\UYC)$ yields
				$(i\succ j)\psi\in\Gamma$.
			\end{itemize}
		\end{proof}


		\begin{theorem}[Completeness for $\SCUzero$]\label{thm:ku-strong}
			Fix $\kappa\geqslant2$.
			For $\Sigma\cup\{\varphi\}\subseteq\ELBU^\kappa$,
		$\Sigma\models\varphi$ implies $\Sigma\vdash_{\SCUzero}\varphi$.

		\end{theorem}
		\begin{proof}
			Assume $\Sigma\nvdash_{\SCUzero}\varphi$.
			Then $\Sigma\cup\{\neg\varphi\}$ is $\SCUzero$-consistent.
			Extend it to a maximal $\SCUzero$-consistent set $\Gamma\in\Omega$.
			By Lemma~\ref{lem:canonical-realization}, choose $\mathcal F,\vec f$ with
			$w:=\langle\Gamma,\mathcal F,\vec f\rangle\in W^c$.
			By Lemma~\ref{lem:truth-ku},
			$\M^c,w\vDash\Sigma$
			and
			$\M^c,w\vDash\neg\varphi$.
			So $\Sigma\not\models\varphi$.
		\end{proof}

	Before stating decidability, we isolate the effective assumption on
	$\Lambda$ used in the finite search below. As in the finitary-model strategy of
	justification logic, what is needed is effective control of generated evidence,
	not merely decidable membership in the constant specification
	\cite{Studer2012-JL}. Say that $\Lambda$ is
	\emph{$\kappa$-effective} if the following can be done effectively.
	Given finite $C\subseteq\ELBU^\kappa$, finite $J\subseteq I$, and finite
	sets $X_i$ $(i\in J)$ of pairs $\langle s,\psi\rangle$,
	with $s$ a graded explanation term and $\psi\in C$, one can decide, for
	each $i\in J$, each $\phi\in C$, and each $1\leqslant m\leqslant\kappa$,
	whether there exists a term $t$ such that
	$\gr(t)\geqslant m$ and $\langle t,\phi\rangle\in
	\bigl(X_i\cup\{\langle c,\lambda\rangle\mid \lambda\in\Lambda\}\bigr)_i^\infty$, where the iterated closure is as in Definition~\ref{def:iter-closure}.
	Simple effective finite-schema choices of $\Lambda$ satisfy this condition; for
	example, $\Lambda$ may be generated by the schemata
	$\phi\wedge\psi\to\phi$ and $\phi\wedge\psi\to\psi$.
	
	\begin{theorem}[Decidability of fixed-level fragments]\label{thm:dec-kappa}
		Fix $\kappa\geqslant2$ and assume that $\Lambda$ is $\kappa$-effective.
		Then satisfiability and validity over $\ELBU^\kappa$ are decidable.
	\end{theorem}
	\begin{proof}[Proof sketch]
		We prove an effective finite-model property. Given $\varphi\in\ELBU^\kappa$, form the finite closure $C$ of $\varphi$ under subformulas, negation, downward closure for understanding levels ($\UO_i^{m+1}\psi$ brings in $\UO_i^m\psi$), the	$\K_i\psi$ associated with $\UO_i^m\psi$, and, for each comparative subformula $(i\succ j)\psi$, the formulas $\K_j\psi,\UO_i^m\psi,\UO_j^m\psi$ for $1\leqslant m\leqslant\kappa$. Only proposition letters and agents occurring in $C$ are relevant.
		
		Call $X\subseteq C$ a complete $C$-type if it decides every formula in $C$
		and respects the Boolean connectives. Enumerate finite structures whose
		states are complete $C$-types, with equivalence relations $R_i$ matching the
		$\K_i$-formulas. For $\UO$-formulas, require class-uniformity: if $X R_i Y$,
		then $X$ and $Y$ contain the same formulas $\UO_i^m\psi$ from $C$. Let
		$\delta_i(X,\psi)$ be the largest such $m$, or $0$ if none exists.	For each $\delta_i(X,\psi)>0$, add one fresh witness
		$e_{i,[X]_i,\psi}$ of that grade, shared across the $R_i$-class $[X]_i$.
		Using $\kappa$-effectiveness, decide whether the least store generated from
		$\Lambda$ and these shared witness pairs contains a pair
		$\langle t,\psi\rangle$ with $\gr(t)\geqslant m$. Accept exactly those
		structures in which the $\UO_i^m\psi$-labels match this generated-witness
		condition and the truth of $\psi$ throughout the relevant $R_i$-class, and in which $(i\succ j)\psi\in X$ iff $\delta_i(X,\psi)>\delta_j(X,\psi)$ and $\K_j\psi\in X$.

		Filtration of any model satisfying $\varphi$ through $C$ yields an accepted
		finite structure, and any accepted structure is realized as a finite model with its states, relations, atom valuation, and least generated explanation stores. Induction on formulas in $C$ gives agreement between truth and membership in the corresponding type. Hence $\varphi$ is
		satisfiable iff some accepted finite structure contains $\varphi$. Since finitely many such structures are checked and all checks are effective, satisfiability is decidable. Validity follows.
	\end{proof}
	
			For the system $\SCUi$, the remaining key step is to build
			maximal consistent sets closed under all $\omega$-introduction instances. We use a countable Lindenbaum--Henkin construction with finite failure witnesses, along the lines of
			\cite{doder2024}.

			\begin{lemma}[$\omega$-Lindenbaum extension]\label{lem:lindenbaum-omega}
				Every $\SCUi$-consistent set extends to a maximal $\SCUi$-consistent set that is
				closed under all $\omega$-introduction instances.
			\end{lemma}
			
			\begin{proof}
				Assume $\Sigma$ is $\SCUi$-consistent.
				Note that $\ELBU$ and the set of $\omega$-introduction instances are countable.
				Enumerate the formulas of $\ELBU$ as $(\theta_s)_{s\in\mathbb N}$.
				Let $\mathcal R$ be the set of all $\omega$-introduction instances, and fix a
				listing $r:\mathbb N\to\mathcal R$ in which every instance occurs infinitely
				often. At stage $s$, write
				$r(s)$ as the rule with premises $\chi_s\to B_s^n$ $(n\in\mathbb N^{+})$
				and conclusion $\chi_s\to B_s^\omega$, where $B_s^\omega$ is obtained from
				$B_s^n$ by replacing the displayed $\UO^n$ with $\UO^\omega$.
				We define the sequence $(\Gamma_s)_{s\in\mathbb N}$ recursively by
			$	\Gamma_0:=\Sigma$, and, given $\Gamma_s$, first set
				\[
			H_s:=
			\begin{cases}
				\Gamma_s\cup\{\theta_s\}, & \Gamma_s\cup\{\theta_s\}\text{ is }\SCUi\text{-consistent},\\
					\Gamma_s\cup\{\neg\theta_s\}, & \text{otherwise}.
				\end{cases}
				\]
				If $\neg(\chi_s\to B_s^\omega)\in H_s$, let
				\[
				\Gamma_{s+1}:=
				H_s\cup\{\neg(\chi_s\to B_s^{n_s})\},
				\]
				where $n_s$ is the least positive integer such that
				$H_s\cup\{\neg(\chi_s\to B_s^{n_s})\}$ is $\SCUi$-consistent.
				Otherwise let $\Gamma_{s+1}:=H_s$.

			We first show that this recursion is well-defined and that every $\Gamma_s$ is
			$\SCUi$-consistent. The basis is immediate. For the step, suppose $\Gamma_s$ is
			consistent. At least one of $\Gamma_s\cup\{\theta_s\}$ and
			$\Gamma_s\cup\{\neg\theta_s\}$ is consistent; otherwise, Lemma~\ref{lem:scui-neg-cons}
			gives $\Gamma_s\vdash_{\SCUi}\theta_s$, while inconsistency of
			$\Gamma_s\cup\{\theta_s\}$ implies inconsistency of
			$\Gamma_s\cup\{\neg\neg\theta_s\}$ by tautology and Lemma~\ref{lem:scui-cut}, so another
			application of Lemma~\ref{lem:scui-neg-cons} gives
			$\Gamma_s\vdash_{\SCUi}\neg\theta_s$, contradiction. Thus $H_s$ is consistent.
				If no such $n_s$ existed in this finite-failure-witness step, then for every $n\in\mathbb N^{+}$,
				$H_s\cup\{\neg(\chi_s\to B_s^{n})\}$ would be inconsistent; by
				Lemma~\ref{lem:scui-neg-cons},
				$H_s\vdash_{\SCUi}\chi_s\to B_s^{n}$ for all $n$, and
				the $\omega$-introduction instance $r(s)$ would yield
				$H_s\vdash_{\SCUi}\chi_s\to B_s^{\omega}$, contradicting consistency of
				$H_s$. 
				
				Let $	\Gamma^*:=\bigcup_{s\in\mathbb N}\Gamma_s$. The construction has the
				following finite-failure property: if the conclusion $\chi\to B^\omega$ of an
				$\omega$-introduction instance has its negation in $\Gamma^*$, then
				$\neg(\chi\to B^n)\in\Gamma^*$ for some $n\in\mathbb N^{+}$.
				Indeed, once $\neg(\chi\to B^\omega)$ has entered the construction, choose a
				later stage $s$ at which $r(s)$ is that same instance. Such a stage exists
				because every instance occurs infinitely often. The witness step at stage $s$
				then adds $\neg(\chi\to B^n)$ for some $n$.

			It remains to check that $\Gamma^*$ is a maximal $\SCUi$-consistent set and is closed under all $\omega$-introduction instances.
			First, $\Gamma^*$ decides every formula: if $\theta=\theta_s$, then stage
			$s+1$ puts either $\theta$ or $\neg\theta$ into $\Gamma^*$.

			Every $\SCUi$-theorem belongs to $\Gamma^*$. Otherwise, since $\Gamma^*$
			decides formulas, $\neg\eta\in\Gamma^*$ for some theorem $\eta$. Then
			$\neg\eta\in\Gamma_s$ for some $s$, and the theorem proof of $\eta$ together
			with the tautology $\eta\to(\neg\eta\to\bot)$ makes $\Gamma_s$ inconsistent.

			The set $\Gamma^*$ is closed under $(\MP)$ in the usual way. It is also closed
			under $\omega$-introduction: if all premises $\chi\to B^n$ of an instance are
			in $\Gamma^*$ but the conclusion $\chi\to B^\omega$ is not, then
			$\neg(\chi\to B^\omega)\in\Gamma^*$. By the finite-failure property,
			$\neg(\chi\to B^n)\in\Gamma^*$ for some $n$. Since the construction is
			increasing, some $\Gamma_s$ contains both $\chi\to B^n$ and its negation,
			contradicting consistency of $\Gamma_s$.

			Finally, $\Gamma^*$ is $\SCUi$-consistent. Otherwise take a proof tree of
			$\bot$ from $\Gamma^*$. By induction on the tree, every node label belongs to
			$\Gamma^*$: assumptions by definition; axioms and $(\NE)$-nodes because
			theorems belong to $\Gamma^*$; $(\MP)$- and $\omega$-introduction nodes by
			the closure just proved; and $(\Nn)$-nodes because their premise subtrees are
			closed and hence prove theorems. Thus $\bot\in\Gamma^*$, so
			$\bot\in\Gamma_s$ for some $s$, contradicting consistency of $\Gamma_s$.
		\end{proof}
		
		Now let $\Omega_\infty$ be the set of maximal $\SCUi$-consistent sets closed under
		all $\omega$-introduction instances.
		\begin{lemma}\label{lem:canonical-realization-omega}
	Let $W_\infty^c$ be the world set obtained by rebuilding
			Definition~\ref{def:cm-xu} with $\Omega$ replaced by $\Omega_\infty$,
			$\ELBU^\kappa$ replaced by $\ELBU$, and every bounded index range
			$\{1,\dots,\kappa\}$ replaced by $\mathbb N^{+}$.
			For every $\Gamma\in\Omega_\infty$, there exist
			$\mathcal F=\{F_i\mid i\in I\}$ and a witness family $\vec f$ such that
		$\langle\Gamma,\mathcal F,\vec f\rangle\in W_\infty^c$.
		\end{lemma}
		\begin{proof}
			Exactly as in Lemma~\ref{lem:canonical-realization}: define
			$f_i^n(\psi):=\psi^n$ on $\{\psi\mid \UO_i^n\psi\in\Gamma\}$,
			and let each $F_i$ be the least closure of the corresponding base under
			\textup{(App$_i$)}, \textup{(Sum$_i$)}, and \textup{(Int$_i$)}.
			The construction is purely syntactic and independent of whether
			$\Gamma\in\Omega$ or $\Gamma\in\Omega_\infty$.
		\end{proof}
		
		Let
	$
				\M_\infty^c=(W_\infty^c,\{R_i^c\}_{i\in I},V^c,E^c,\gr^c,\{\Ex_i^c\}_{i\in I})$	be the canonical model over $\Omega_\infty$ obtained in this way.
				\begin{lemma}[Full $\K_i$-Existence Lemma]\label{lem:k-full}
					If $w=\langle \Gamma, \mathcal F, \vec f\rangle\in W_\infty^c$ and
					$\neg \K_i\phi \in \Gamma$, then there exists
					$w'=\langle\Delta,\mathcal G,\vec g\rangle\in W_\infty^c$ such that
					$wR_i^c w'$ and $\neg\phi\in \Delta$.
				\end{lemma}
			\begin{proof}
				Let
				$\Delta^-:=\{\psi\mid \K_i\psi\in\Gamma\}\cup\{\neg\phi\}$.
				We first show that $\Delta^-$ is $\SCUi$-consistent. Suppose otherwise.
				Let a proof tree of $\bot$ from $\Delta^-$ be given. We claim, by induction
				on this tree, that for every node label $\alpha$,
				$\K_i(\neg\phi\to\alpha)\in\Gamma$.
				
					For assumption leaves, either $\alpha=\psi$ with $\K_i\psi\in\Gamma$, or
					$\alpha=\neg\phi$. In the first case, normal modal reasoning gives
					$\K_i(\neg\phi\to\psi)\in\Gamma$ from $\K_i\psi\in\Gamma$; in the second,
					$\K_i(\neg\phi\to\neg\phi)\in\Gamma$ follows by necessitation. Axiom leaves,
					$(\NE)$-nodes, and closed $(\Nn)$-nodes are theorems, so the desired formula
					$\K_i(\neg\phi\to\alpha)$ also follows by necessitation. The $(\MP)$ case is standard.
				
				For an $\omega$-introduction node, write its premises as
				$\chi\to B_n$ $(n\in\mathbb N^{+})$ and its conclusion as
				$\chi\to B_\omega$.
				By the induction hypothesis,
				$\K_i(\neg\phi\to(\chi\to B_n))\in\Gamma$ for every $n$.
				Equivalently, by normal modal reasoning,
				$\K_i((\neg\phi\wedge\chi)\to B_n)\in\Gamma$ for every $n$.
					By propositional reasoning, these give the premises
					$\top\to \K_i((\neg\phi\wedge\chi)\to B_n)$ of the corresponding
					$(\omega \mathrm{I})$-instance, whose conclusion is
					$\top\to\K_i((\neg\phi\wedge\chi)\to B_\omega)$. Since $\Gamma$ is closed under all
				$\omega$-introduction instances, we obtain
				$\top\to\K_i((\neg\phi\wedge\chi)\to B_\omega)\in\Gamma$, hence
				$\K_i((\neg\phi\wedge\chi)\to B_\omega)\in\Gamma$, and therefore
				$\K_i(\neg\phi\to(\chi\to B_\omega))\in\Gamma$.
				
				At the root we obtain $\K_i(\neg\phi\to\bot)\in\Gamma$, and hence
				$\K_i\phi\in\Gamma$, contradicting $\neg\K_i\phi\in\Gamma$.
				Thus $\Delta^-$ is consistent. By Lemma~\ref{lem:lindenbaum-omega}, extend
				$\Delta^-$ to some $\Delta\in\Omega_\infty$.
				
				For every $n\in\mathbb N^{+}$ and formula $\chi$,
				$\UO_i^n\chi\in\Gamma$ iff $\UO_i^n\chi\in\Delta$: the forward direction uses
				$(\4^*)$, and the backward direction uses $(\5^*)$, exactly as in the bounded
				$\K_i$-Existence Lemma. Hence we may set $g_i^n:=f_i^n$ for all $n$.
				For $j\neq i$, let $g_j^n(\psi):=\psi^n$ whenever $\UO_j^n\psi\in\Delta$.
				Build $\mathcal G=\{G_j\mid j\in I\}$ from $\Delta$ and $\vec g$ by the same
				base-and-closure construction used in Lemma~\ref{lem:canonical-realization-omega}.
				Then $w':=\langle\Delta,\mathcal G,\vec g\rangle$ belongs to $W_\infty^c$.
				Since $\{\psi\mid \K_i\psi\in\Gamma\}\subseteq\Delta$ and $g_i^n=f_i^n$ for all
				$n$, we have $wR_i^c w'$. Finally, $\neg\phi\in\Delta$ by construction.
				\end{proof}

				\begin{lemma}\label{lem:finite-u-full}
					Let $w=\langle\Gamma,\mathcal F,\vec f\rangle\in W_\infty^c$, fix
					$i\in I$ and $n\in\mathbb N^{+}$, and assume
					$\UO_i^n\phi\notin\Gamma$ and $\langle t,\phi\rangle\in F_i$ for some
					$t\in E_n^c$.
					There exists $w'=\langle\Delta,\mathcal G,\vec g\rangle\in W_\infty^c$ such that
					$wR_i^c w'$ and $\langle t,\phi\rangle\notin G_i$.
				\end{lemma}
				\begin{proof}
			It is the same as Lemma~\ref{lem:uo-refute-transfer}, with
			$\mathbb N^{+}$ in place of $\{1,\dots,\kappa\}$. $\omega$-introduction is not used.
				\end{proof}
			\begin{lemma}[Truth Lemma for $\ELBU$]\label{lem:truth-full}
				For $\varphi\in\ELBU$ and 
				$w=\langle\Gamma,\mathcal F,\vec f\rangle\in W_\infty^c$:
			$\M_\infty^c,w\vDash\varphi
			$ iff $
				\varphi\in\Gamma$.
		
		\end{lemma}
			\begin{proof}
					Induction on $\varphi$.
				The Boolean cases are standard. The $\K_i$ case uses
				Lemma~\ref{lem:k-full}. The finite-level $\UO_i^n$ case is the same as in
				Lemma~\ref{lem:truth-ku}, using Lemma~\ref{lem:finite-u-full} for the
				refutation direction.
			
				 For $\varphi=\UO_i^\omega\psi$.
						\begin{itemize}
					\item If $\UO_i^\omega\psi\in\Gamma$, then by $(\IYB_\omega^n)$,
					$\UO_i^n\psi\in\Gamma$ for every $n\geqslant1$. By IH,
					$\M_\infty^c,w\vDash\UO_i^n\psi$
					for every $n\geqslant1$, hence
					$\M_\infty^c,w\vDash\UO_i^\omega\psi$.
					\item If $\M_\infty^c,w\vDash\UO_i^\omega\psi$,
					then by semantics
					$\M_\infty^c,w\vDash\UO_i^n\psi$
					for all $n\geqslant1$, so by the finite-level $\UO$-case
					$\UO_i^n\psi\in\Gamma$ for all $n\geqslant1$.
						By closure under the basic $\omega$-introduction instance,
						$\UO_i^\omega\psi\in\Gamma$.
				\end{itemize}
				
					 For $\varphi=(i\succ j)\psi$.
					The case $i=j$ is the same as in Lemma~\ref{lem:truth-ku}. Assume $i\neq j$.
					\begin{itemize}
							\item Assume $(i\succ j)\psi\in\Gamma$.
							By $(\CMP1)$ and $(\CMP\omega)$,
							$\K_j\psi,\UO_i^1\psi,\neg\UO_j^\omega\psi\in\Gamma$; hence
							$\M_\infty^c,w\vDash\K_j\psi$ by the preceding $\K$-case.
							Let $S_j:=\{n\geqslant1\mid \UO_j^n\psi\in\Gamma\}$.
						By $(\IYB)$, $S_j$ is downward closed and $S_j\neq\mathbb N^{+}$; otherwise the basic
							$\omega$-introduction instance would contradict
							$\neg\UO_j^\omega\psi\in\Gamma$.
							If $S_j=\emptyset$, then $\deg_j(w,\psi)=0$, while
							$\UO_i^1\psi\in\Gamma$ gives $\deg_i(w,\psi)\geqslant1$ by the	finite-level $\UO$-case.
							If $S_j\neq\emptyset$, let $m+1$ be the least not in $S_j$.
							Then $m\geqslant1$ and $S_j=\{1,\dots,m\}$.
							Hence $\UO_j^m\psi,\neg\UO_j^{m+1}\psi\in\Gamma$ and, by $(\CMPn)$,
							$\UO_i^{m+1}\psi\in\Gamma$.
						By the finite-level $\UO$-case,
						$\deg_j(w,\psi)=m$ and $\deg_i(w,\psi)\geqslant m+1$.
						In either case, $\deg_i(w,\psi)>\deg_j(w,\psi)$, so
						$\M_\infty^c,w\vDash(i\succ j)\psi$.
							\item Assume
							$\M_\infty^c,w\vDash(i\succ j)\psi$.
							Then $\K_j\psi\in\Gamma$ by the $\K$-case. The degree inequality
							also gives $\deg_j(w,\psi)\neq\omega$.
							If $\deg_j(w,\psi)=0$, then
							$\deg_i(w,\psi)\geqslant1$ and $\UO_j^1\psi$ is false. By the finite-level $\UO$-case, $\UO_i^1\psi,\neg\UO_j^1\psi\in\Gamma$; hence
							$(i\succ j)\psi\in\Gamma$ by $(\UYC0)$.
							If $\deg_j(w,\psi)=m\geqslant1$, then
							$\deg_i(w,\psi)\geqslant m+1$. By the finite-level $\UO$-case,
							$\UO_j^m\psi,\neg\UO_j^{m+1}\psi, \UO_i^{m+1}\psi \in\Gamma$; then $(\UYC)$ yields
							$(i\succ j)\psi\in\Gamma$.
					\end{itemize}
			
		\end{proof}
			\begin{theorem}[Strong completeness for $\SCUi$]\label{thm:strong-full-inf}
			For all $\Sigma\cup\{\varphi\}\subseteq\ELBU$,
			$	\Sigma\models\varphi
			$ implies $
			\Sigma\vdash_{\SCUi}\varphi$.
			
		\end{theorem}
			\begin{proof}
				Assume $\Sigma\nvdash_{\SCUi}\varphi$.
				By Lemmas~\ref{lem:scui-neg-cons} and~\ref{lem:lindenbaum-omega}, choose
				$\Gamma^\infty\in\Omega_\infty$ with
				$\Sigma\cup\{\neg\varphi\}\subseteq\Gamma^\infty$.
				By Lemma~\ref{lem:canonical-realization-omega}, choose $\mathcal F,\vec f$
				such that $w:=\langle\Gamma^\infty,\mathcal F,\vec f\rangle\in W_\infty^c$.
				Let $\M_\infty^c$ be the corresponding canonical model over $\Omega_\infty$.
				By Lemma~\ref{lem:truth-full},
				$\M_\infty^c,w\vDash\Sigma$ and $\M_\infty^c,w\vDash\neg\varphi$.
				Thus $\Sigma\not\models\varphi$.
		\end{proof}

		\section{Conclusion}\label{sec:conclusion}
		We have developed a comparative epistemic logic of understanding. 
	 	On the semantic side, the framework is based on agent-indexed graded explanations;
		on the proof-theoretic side, it separates a bounded finitary layer from a full
		infinitary layer. In this way, the logic captures understanding in degrees and
		comparative understanding with respect to the same formula at issue, while also supporting
		strong completeness for the intended calculi and decidability for each fixed
		finite-level fragment.

		The framework also suggests several natural directions for further work.
		On the technical side, it would be worth studying richer proof-theoretic
		presentations for the full language. 
		On the dynamic side, one may ask how graded and comparative understanding behaves
		under information change, communication, or learning. More broadly, the present
		logic invites closer comparison both with other comparative epistemic frameworks
		and with neighboring formal accounts of explanation and understanding.

		\paragraph{Acknowledgements.}
		The author thanks the anonymous reviewers for their helpful suggestions that led to
		many improvements. The author also thanks Qiang Wang for commenting on a very early
		draft of this paper. This work is supported by the grant 24CZX085 from the National
		Social Science Fund of China.
		
		\bibliographystyle{eptcs}
		\bibliography{generic}
	\end{document}